\documentclass{amsart}
\pdfoutput=1

\usepackage{amssymb}
\usepackage{amsthm}
\usepackage{amsfonts}
\usepackage{amsmath}
\usepackage{pmboxdraw}
\usepackage{verbatim} 
\usepackage{graphicx}
\usepackage{color}
\usepackage[colorlinks=true, citecolor=blue, filecolor=black, linkcolor=black, urlcolor=black]{hyperref}
\usepackage{cite}
\usepackage[normalem]{ulem}
\usepackage{subcaption}
\usepackage{bbm}
\usepackage{bm}
\usepackage{mathtools}
\usepackage{kantlipsum}
\allowdisplaybreaks
\usepackage{enumitem}
\usepackage{hyperref}
\usepackage{tikz}
\usepackage{pgfplots}
\pgfplotsset{compat=1.18}
\usepackage{todonotes}

\hypersetup{
    colorlinks=true,
    linkcolor=blue,
    filecolor=magenta,      
    urlcolor=cyan,
    pdfpagemode=FullScreen,
    }

\newcommand{\RN}[1]{%
	\textup{\uppercase\expandafter{\romannumeral#1}}%
}

\def\bfs{\boldsymbol}

\def\EE{\mathcal{E}}

\def\KK{\mathcal{K}}

\def\C{\mathbb{C}}

\def\P{\mathbf{P}}
\def\R{\mathbb{R}}

\newcommand{\re}{\operatorname{Re}}

\newcommand{\ud}{\, \mathrm{d}}
\newcommand{\interior}{\operatorname{Int}}

\theoremstyle{plain}
\newtheorem{thm}{Theorem}[section]

\newtheorem{cor}[thm]{Corollary}
\newtheorem{lem}[thm]{Lemma}

\newtheorem{prop}[thm]{Proposition}

\theoremstyle{remark}
\newtheorem{rem}{Remark}[section]
\newtheorem{defn}{Definition}

\numberwithin{equation}{section}

\begin{document}

\title[Three phases of the Planar Equilibrium Measure]{Three topological phases of \\ the elliptic Ginibre ensembles with a point charge}
\author{Sung-Soo Byun}
\address{Department of Mathematical Sciences and Research Institute of Mathematics, Seoul National University, Seoul 151-747, Republic of Korea}
\email{sungsoobyun@snu.ac.kr}

\author{Eui Yoo}
\address{Department of Mathematical Sciences, Seoul National University, Seoul, Republic of Korea}
\email{yysh0227@snu.ac.kr}

\begin{abstract} 
We consider the complex and symplectic elliptic Ginibre matrices of size $(c+1)N \times (c+1)N$, conditioned to have a deterministic eigenvalue at $ p \in \mathbb{R} $ with multiplicity $ c N $. We show that their limiting spectrum is either simply connected, doubly connected, or composed of two disjoint simply connected components. Moreover, denoting by $\tau \in [0,1]$ the non-Hermiticity parameter, we explicitly characterise the regions in the parameter space $  (p, c, \tau) $ where each topological type emerges. For cases where the droplet is either simply or doubly connected, we provide an explicit description of the limiting spectrum and the corresponding electrostatic energies. As an application, we derive the asymptotic behaviour of the moments of the characteristic polynomial for elliptic Ginibre matrices in the exponentially varying regime. 
\end{abstract}

\maketitle

\section{Introduction and main results}

Despite receiving significant attention in recent years, non-Hermitian random matrix theory has historically been less explored than its Hermitian counterpart. This is partly because many key tools used in Hermitian random matrix theory, such as classical orthogonal polynomial theory and group integral techniques, cannot generally be applied to non-Hermitian random matrices. Nonetheless, the past two decades have seen remarkable progress in non-Hermitian random matrix theory, aided by deep connections to other mathematical areas such as the theory of Coulomb gases \cite{Se24,Fo10,Le22}. We refer the reader to \cite{BF24} for a recent review of the progress in the field of non-Hermitian random matrices.

Not only is it more challenging, but non-Hermitian random matrix theory has also been found to exhibit more fruitful features than Hermitian random matrix theory. One prominent example is its connection to the topological and conformal geometric properties of the limiting spectral distribution, often referred to as the droplet. For instance, the work of Jancovici et al. \cite{JMP94,TF99} in the 1990s introduced the surprising observation that the precise asymptotic behaviour of the free energies is intricately linked to the topological properties of their droplets, as represented by the Euler characteristics (see also \cite{CFTW15}).
Furthermore, recent studies have revealed that the behaviour of these ensembles depends in a highly non-trivial way on the multiple connectivity or the number of disjoint connected components of droplets \cite{BSY24,BKS23,BP24,BKSY25,BFL25, Ch22,Ch23,ACCL24,ACC23,BY23,AC24,ACC23a}. For the critical case involving certain singularities, see \cite{KLY23,CEJ24,BLY21,LZ23,SS22,JV23,BGM18,CL24,CL24a} and references therein.

This, in turn, calls for explicit derivations of droplets that naturally arise in non-Hermitian random matrices, particularly those with rich topological structure. In this direction, two natural models have been actively investigated in the field, both constructed from the Ginibre matrix \cite{BF24}, a random matrix with independent and identically distributed Gaussian entries, or its variants. 
The first model adopts an electrostatic perspective. In this approach, a non-trivial point charge is imposed, or equivalently, one considers conditional Ginibre matrices with a prescribed deterministic eigenvalue, see e.g.  \cite{BBLM15,BSY24,KLY23,CK22,BKP23,KKL24,BGM18,LY17} and references therein. 
The second model takes a more matrix-theoretic approach, involving the addition of a deterministic matrix to the Ginibre matrix, leading to what are known as deformed Ginibre matrices, see e.g. \cite{CEJ24,CCEJ24,EJ23,LZ22,LZ23} and references therein. 

In this work, we take the first approach and investigate the limiting spectrum of elliptic Ginibre matrices with a point charge. The elliptic Ginibre matrices are indexed by a non-Hermiticity parameter and interpolate between the Ginibre matrices and Gaussian Hermitian random matrices---in our case, the Gaussian unitary and symplectic ensembles. For Ginibre matrices with a point charge, the associated droplet is characterised in the seminal work \cite{BBLM15}, where it was shown that the droplet is either simply or doubly connected. Our main results in this paper extend these findings, revealing that, when the non-Hermiticity parameter is considered, an additional third phase emerges: a regime where the droplet consists of two connected components. We explicitly derive the regions in the parameter space \((p, c, \tau)\) where each topological type arises. Furthermore, when the droplet is either simply or doubly connected, we provide an explicit description of it as well as its electrostatic energies.  
As a consequence, we derive the asymptotic behaviours of the moments of the characteristic polynomials of the elliptic Ginibre matrices.

\medskip 

Let us now be more precise in introducing our results. 
We consider configurations of points $\bfs{z}=\{z_j\}_{j=1}^N$ in the complex plane, with joint probability distribution functions
\begin{align}
\ud\P_N^{ \mathbb C }(\boldsymbol{z}) & = \frac{1}{Z_N^{ \mathbb{C} }(W) } \prod_{ 1 \le j<k \le N } |z_j-z_k|^2 \prod_{j=1}^N e^{ -N W (z_j) }   \ud A(z_j), \label{Gibbs complex} 
\\
\ud\P_N^{ \mathbb H }(\boldsymbol{z}) & = \frac{1}{Z_N^{ \mathbb{H} }(W) } \prod_{1 \le j<k \le N} |z_j-z_k|^2 \prod_{1 \le j \le k \le N} |z_j-\overline{z}_k|^2 \prod_{j=1}^N   e^{ -2NW(z_j) }   \ud A(z_j), \label{Gibbs symplectic}
\end{align}
where $\ud A(z)=d^2z/\pi$ is the area measure. Here $W:\mathbb{C} \to \mathbb{R}$ is a given external potential, and ${Z_N^{ \mathbb{C} }(W) } $ and $ {Z_N^{ \mathbb{H} }(W) }  $ are the partition functions. The ensembles \eqref{Gibbs complex} and \eqref{Gibbs symplectic} are known as the random normal matrix ensemble and the planar symplectic ensemble, respectively. Moreover, they are equivalent to two-dimensional Coulomb gases at inverse temperature $\beta=2$, with Dirichlet and Neumann boundary conditions, respectively. We also refer to \cite{LMS19,FJ96,DLM19} and references therein for a realisation as a fermionic system.

The limiting distribution of the point process $\bfs z$ can be effectively described using the logarithmic potential theory.  
Let us briefly recall some basic notions and properties from potential theory, see \cite{ST97} for a comprehensive source. 
For a given probability measure $\mu$, the weighted logarithmic energy is given by 
\begin{align} \label{def of log energy}
    I_W(\mu):=\iint_{\mathbb{C}^2}\log \frac{1}{|z- w|}\ud \mu(z)\ud \mu(w) + \int_\mathbb{C}W(z) \ud \mu(z).
\end{align}
It is well known that for a general admissible potential $W$, there exists a unique measure $\mu_W$ that minimises $I_W$. 
Furthermore, $\mu_W$ is characterised by the variational conditions (Euler-Lagrange equations)
\begin{equation} \label{eq:variational}
  \int_\mathbb{C} \log\frac{1}{|z-w|  }\ud \mu_W(w) + \frac{1}{2}W(z)  \begin{cases}
  = C_W &   z \in \text{supp }\mu_W, 
    \smallskip   
    \\
\ge C_W &  z \in \mathbb{C}. 
  \end{cases}
\end{equation}
Here, $C_W$ is called the (modified) Robin's constant. 
From the structural point of view, Frostman's theorem asserts that $\mu_W$ is absolutely continuous with respect to the area measure $\ud A$, and takes the form
\begin{align} \label{eq for Frostman}
    d\mu_W = \Delta W \cdot \mathbbm{1}_{S_W}\ud A, \qquad (\Delta :=\partial \bar{\partial}), 
\end{align}
where $S_W$ is a certain compact subset of the complex plane called the droplet.

The equilibrium measure $\mu_W$ is closely related to the ensembles \eqref{Gibbs complex} and \eqref{Gibbs symplectic}. By standard equilibrium convergence, the empirical measure $\frac{1}{N} \sum_{j=1}^N \delta_{z_j}$ of the point process $\bfs z$ converges to the equilibrium measure $\mu_W$, see e.g. \cite{Se24,CGZ14}. In order to see this more intuitively, notice that the Gibbs measures \eqref{Gibbs complex} and \eqref{Gibbs symplectic} are proportional to $ \exp( - {\rm{H}}_N^{\C}( \bfs{z} ) )$ and $ \exp( - {\rm{H}}_N^{\mathbb{H}}( \bfs{z} ) )$, where the Hamiltonians are given by 
\begin{align}
{\rm{H}}_N^{\mathbb{C}}( \bfs{z} ) &= \sum_{ 1\le j<k \le N } \log \frac{1}{|z_j-z_k|^2} +N \sum_{j=1}^N W(z_j),  \label{Ham complex}
\\
{\rm{H}}_N^{\mathbb{H}}( \bfs{z} ) &= \sum_{  1\le j<k \le N  } \log \frac{1}{|z_j-z_k|^2   }+ \sum_{ 1\le j \le k \le N } \log \frac{1}{ |z_j-\overline{z}_k|^2 } +2N \sum_{j=1}^N W(z_j).  \label{Ham symplectic}
\end{align}
Thus one can see that $I_W$ in \eqref{def of log energy} corresponds to the continuum limit of these Hamiltonians, after taking proper normalisations. Here, it has been assumed that $W(z)=W(\overline{z})$ for the second case.  
From the equilibrium convergence, when investigating the macroscopic distribution of the Coulomb gas ensembles, one of the key tasks is to solve the equilibrium measure problem. That is, for a given potential $W$, one aims to determine the associated equilibrium measure $\mu_W$. This constitutes a particular type of inverse problem, and due to the structure in \eqref{eq for Frostman}, the main step in this problem is to identify the droplet $S_W$.
We refer the reader to \cite{By24,BS20,ABK21,BK12,BM15,BBLM15,Ch23a,BFL25,LD21,Ad18} and references therein for recent development on the planar equilibrium measure problem.

\medskip 

We now turn to our particular model of interest, the conditional elliptic Ginibre ensembles. In order to introduce this, let us first write $G$ for the Ginibre matrices whose entries are complex or quaternionic Gaussian random variables with mean zero and variance $1/N$. Introducing a non-Hermiticity parameter $\tau \in [0,1]$, the elliptic Ginibre matrices are then defined by
\begin{equation} \label{def of eGinibre}
X_{\tau}:= \frac{\sqrt{1+\tau}}{2}(G+G^*)+\frac{\sqrt{1-\tau}}{2}(G-G^*).
\end{equation}
Then its eigenvalue distribution follows \eqref{Gibbs complex} and \eqref{Gibbs symplectic} respectively, where the associated potential is given by 
\begin{equation} \label{def of potential eGinibre}
W^{ \rm e }(z)= \frac{1}{1-\tau^2}\Big( |z|^2-\tau \re z^2 \Big). 
\end{equation} 
Furthermore, as $N \to \infty$, the eigenvalues tend to be uniformly distributed within an ellipse 
\begin{equation} \label{def of elliptic law}
\mathsf{E} = \Big\{ (x,y)\in \mathbb{R}^2 : \Big( \frac{x}{1+\tau}\Big)^2 + \Big( \frac{y}{1-\tau}\Big)^2 \le 1 \Big\}, 
\end{equation} 
which is often called the elliptic law. 

Next, for a given $c \ge 0$, we consider the elliptic Ginibre matrix of size $(c+1)N \times (c+1)N$, conditioned to have deterministic eigenvalue at $p \in \R$ with multiplicity $cN$. Then the remaining $N$ random eigenvalues again follow the distributions \eqref{Gibbs complex} and \eqref{Gibbs symplectic}, where the associated external potential is given by 
\begin{equation} \label{eq:potential}
    Q(z) = \frac{1}{1-\tau^2}\Big(|z|^2 -\tau \re z^2\Big)- 2c\log|z-p|. 
\end{equation}
Such a logarithmic singularity is often called the point charge insertion or the Fisher-Hartwig singularity. Furthermore, as will be discussed below this section, it is closely related to the moments of the characteristic polynomials \cite{AV03}. The way to construct the random matrix model with a logarithmic point charge is also known as the inducing procedure \cite{FBKSZ12}.

It follows from \eqref{eq for Frostman} that the equilibrium measure $\mu_Q$ is of the form 
\begin{align}\label{eq:eqmeasure}
    \ud\mu_Q  = \frac{1}{1-\tau^2}\mathbbm{1}_{S_Q} \ud A.
\end{align}
In this paper, we aim to provide the topological characterisation of the droplet \(S \equiv S_Q\). For this purpose, we distinguish the parameter space of $(p,c,\tau)$ into three distinct regimes.

\begin{defn}[Regimes of the parameters $p$, $c$ and $\tau$] \label{Def_regimes of p and c} We define the following different regimes, cf. Figure~\ref{fig:phase}.  
\begin{itemize}
    \item (Regime I) The first regime is the most explicit and corresponds to the case where \( p \) and \( c \) lie within the following ranges: 
    \begin{equation}\label{eq_doubly connected range 1}
    p \le \min \Big \{  2\sqrt{\frac{2\tau(1+\tau)}{3+\tau^2}},2 \sqrt{ \frac{     \tau(1-\tau-2c\tau)   }{  1-\tau  } } \Big \}  \quad \text{and} \quad  0\leq c\leq \frac{1-\tau}{2\tau},
    \end{equation}
    or 
    \begin{equation}\label{eq_doubly connected range 2}
  2\sqrt{\frac{2\tau(1+\tau)}{3+\tau^2}} \le   p \le (1+\tau)\sqrt{1+c}-\sqrt{c(1-\tau^2)} 
  \quad \text{and} \quad 0\leq c\leq \frac{(1-\tau)^3}{2\tau(3+\tau^2)}. 
    \end{equation}
  \item (Regime II) The second regime corresponds to the case where for a given $\tau,$ the other parameters $c$ and $p$ are given in terms of two parameters $a$ and $\kappa$ as
\begin{align}
c& \equiv c(a,\kappa) =   \frac{\kappa}{a^2} \frac{(1-a^2)^2 (1-\tau a^2) +  a^2\kappa  }{(1-a^2)^2(1-\tau^2 + 2\tau \kappa ) - \kappa^2 }, \label{eq_simply connected c}
\\
p & \equiv p(a,\kappa) =  \sqrt{ \frac{1+\tau}{1-\tau} }
 \frac{  (1-\tau)(1-a^2)(  1+\tau a^2 ) -( 1-\tau a^2) \kappa   }{a\sqrt{ (1-a^2)^2 (1-\tau^2 + 2\tau \kappa) - \kappa^2 } }.\label{eq_simply connected p}
\end{align}
Here, the parameters $a$ and $\kappa$ lie in the range
\begin{equation}
  a\in(0,1), \qquad \kappa \in [0, \kappa_{ \rm cri}), 
\end{equation}
where $\kappa_{ \rm cri }$ is specified as a unique zero of $H(a, \cdot)$ in \eqref{def of H(a,kappa)}. 
\smallskip 
 \item (Regime III) This corresponds to the case where the ranges of $p, c$ and $\tau$ lie outside the above two regimes.
\end{itemize}
\end{defn}

\begin{figure}[t]
     \begin{subfigure}{0.4\textwidth}
        \begin{center}
            \includegraphics[width=\textwidth]{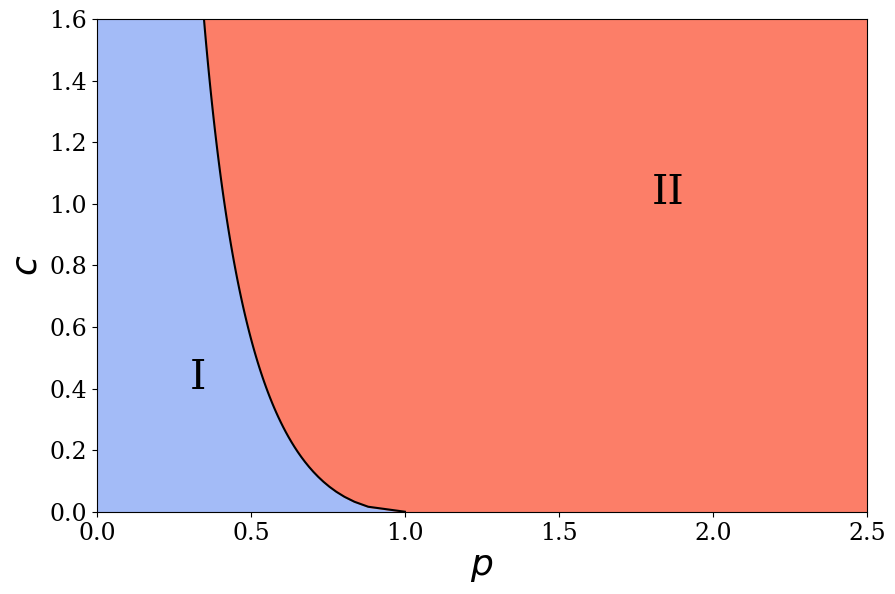}
        \end{center}
    \subcaption{ $\tau=0$ }
    \end{subfigure}
  \qquad     \begin{subfigure}{0.4\textwidth}
        \begin{center}
            \includegraphics[width=\textwidth]{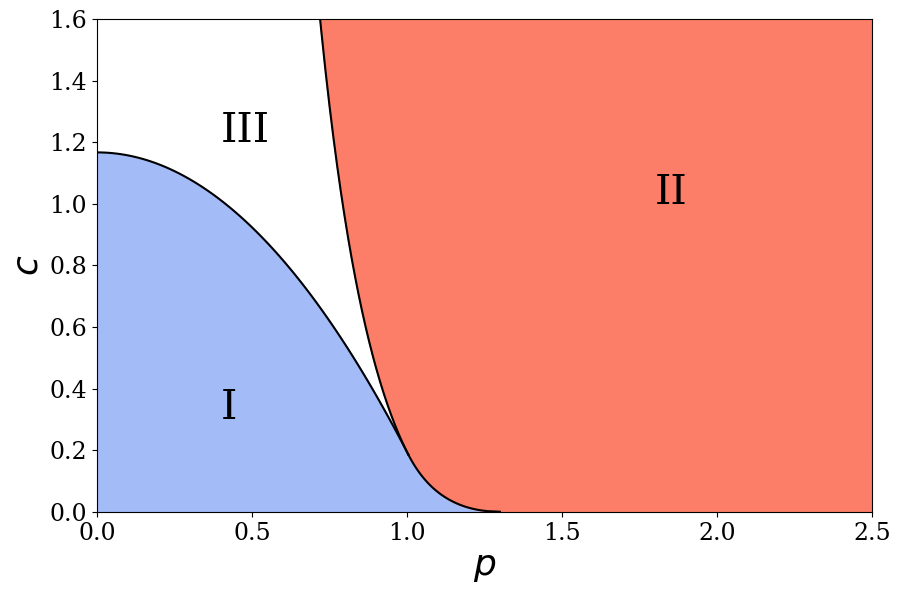}
        \end{center}
    \subcaption{ $\tau=0.3$ }
    \end{subfigure}
  
     \begin{subfigure}{0.4\textwidth}
        \begin{center}
            \includegraphics[width=\textwidth]{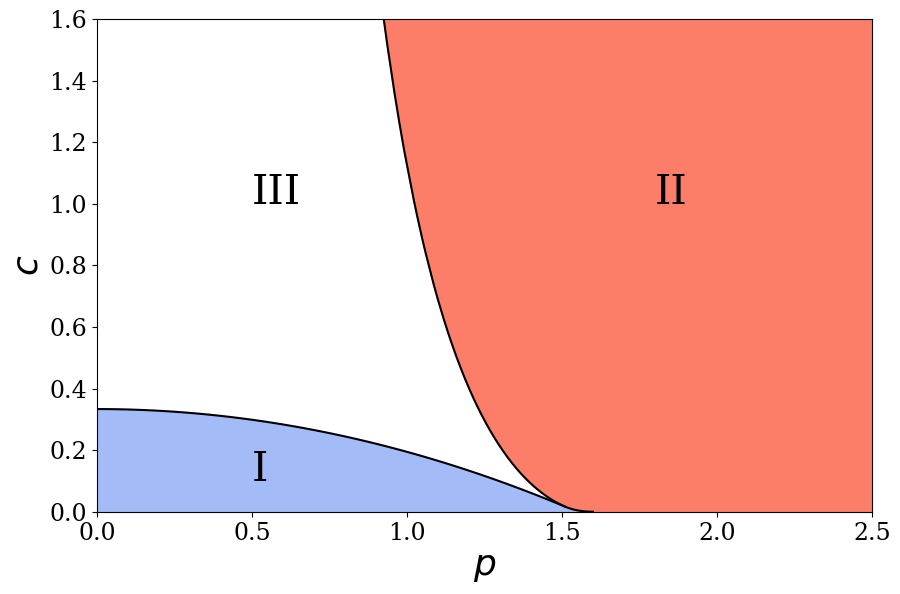}
        \end{center}
    \subcaption{ $\tau=0.6$ }
    \end{subfigure}
     \qquad   \begin{subfigure}{0.4\textwidth}
        \begin{center}
            \includegraphics[width=\textwidth]{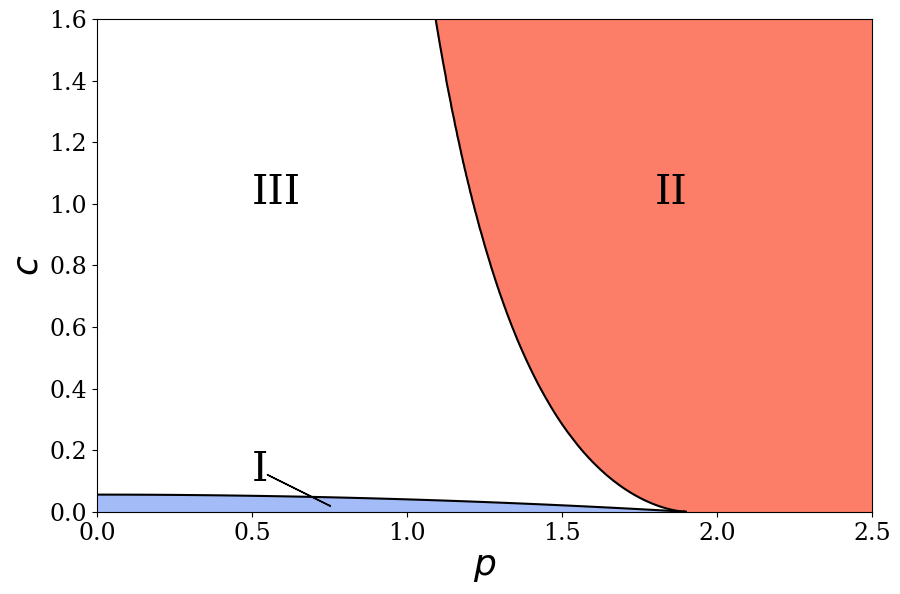}
        \end{center}
    \subcaption{ $\tau=0.9$ }
    \end{subfigure}
    \caption{ The plot illustrates Regimes I, II, and III in Definition~\ref{Def_regimes of p and c}, across various values of $\tau$. Notably,  for $\tau = 0$, only Regimes I and II are present, while for $p = 0$, only Regimes I and III persist, which is consistent with discussions in Remark~\ref{Remark_regimes extremal}. }
    \label{fig:phase}
\end{figure}

Our first main result provides the explicit phase characterisation of the droplet. 
 
\begin{thm}[\textbf{Topological characterisation of the droplet}] \label{Thm_main droplet}
The droplet \( S \) associated with \( Q \) defined in \eqref{eq:potential} is either doubly connected, simply connected, or composed of two disjoint simply connected components. More precisely, we have the following. 
\begin{itemize}
    \item[\textup{(i)}] The droplet is doubly connected  if and only if $(p,c,\tau)$ falls within \textup{Regime I}. 
    \smallskip 
      \item[\textup{(ii)}] The droplet is simply connected if and only if $(p,c,\tau)$ falls within \textup{Regime II}. 
    \smallskip 
     \item[\textup{(iii)}] The droplet consists of two disjoint simply connected components if and only if $(p,c,\tau)$ falls within \textup{Regime III}. 
\end{itemize}
\end{thm}

\begin{rem}[Phases in extremal cases] \label{Remark_regimes extremal}
We compare Theorem~\ref{Thm_main droplet} with known results for two extremal cases.
\begin{itemize}
    \item (The Ginibre case $\tau=0$, cf.\cite{BBLM15}). In this case, Regime III reduces to a null set, leaving only Regimes I and II. These two regimes can be determined by the condition 
    $$
    p \le  \sqrt{1+c}-\sqrt{c}  \qquad \textup{or} \qquad  p > \sqrt{1+c}-\sqrt{c}
    $$
    for Regimes I and II, respectively. This corresponds to the regimes investigated in \cite{BBLM15}.
    \smallskip 
    \item (The point charge at the origin $p=0$, cf. \cite{By24}). In contrast to the previous case, if $p=0$, then Regime II reduces to a null set, leaving only Regimes I and III. These two regimes can be determined by the condition 
    $$
    \tau \le \frac{1}{1+2c}  \qquad \textup{or} \qquad \tau > \frac{1}{1+2c} 
    $$
      for Regimes I and III, respectively. This corresponds to the regimes investigated in \cite{By24}.
\end{itemize}
\end{rem}

Recall that the weighted logarithmic energy is given by \eqref{def of log energy} and the equilibrium measure $\mu_Q$ is of the form \eqref{eq:eqmeasure}. In cases (i) and (ii) of Theorem~\ref{Thm_main droplet}, we further provide an explicit description of the droplets and an evaluation of the logarithmic energies.

\begin{thm}[\textbf{Description of the droplet and electrostatic energies}] \label{Thm_droplet and energy} We have the following. 
\begin{itemize}
    \item[\textup{(i)}] Suppose that $(p,c,\tau)$ falls within \textup{Regime I}. Then the droplet is given by 
    \begin{equation} \label{droplet_doubly connected}
    S= \bigg\{ (x,y) \in \R^2 : \Big( \frac{x}{1+\tau}\Big)^2 + \Big( \frac{y}{1-\tau}\Big)^2 \le 1+c\,,\, (x-p)^2+y^2 \ge c(1-\tau^2) \bigg\}. 
    \end{equation}
    Furthermore, the weighted logarithmic energy is given by $I_Q(\mu_Q) = \mathcal{I}_{ \rm d }(p,c,\tau)$, where
    \begin{equation} \label{weighted energy doubly connected}
    \mathcal{I}_{ \rm d }(p,c,\tau) := \frac{3}{4}+\frac{3c}{2} + \frac{c^2}{2}\log \Big( c(1-\tau^2) \Big) -\frac{(1+c)^2}{2}\log(1+c) - \frac{c\, p^2}{1+\tau}.
    \end{equation}  
      \item[\textup{(ii)}] Suppose that $(p,c,\tau)$ falls within \textup{Regime II}. Then the droplet is given by the closure of the interior of the real-analytic Jordan curve formed by the image of the unit circle under the rational map
    \begin{align}\label{eq_simply conformal map}
    f(z) = R\Big(z + \frac{\tau}{z} - \frac{\kappa}{z-a}-\frac{\kappa}{a(1-\tau)}\Big), 
\end{align}
where $R>0$, $a \in (0,1)$, and $\kappa \in [0, \kappa_{ \rm cri})$. Here, $(R,a, \kappa)$ is a solution to the coupled algebraic equations 
\begin{align}  
    1&  = \frac{R^2 }{1-\tau^2} \Big( 1-\tau^2 +  2\tau\kappa -\frac{\kappa^2}{ (1-a^2)^2 } \Big) , \label{eq:sc2}
    \\
    c&= \frac{R^2 \kappa }{1-\tau^2}\Big(\frac{1-\tau a^2}{a^2} + \frac{\kappa}{(1-a^2)^2}\Big),\label{eq:sc3}
    \\
    p& = \frac{R}{a}\Big(1+\tau a^2-\frac{1-\tau a^2}{1-\tau}\frac{\kappa}{1-a^2}\Big).\label{eq:sc4}
\end{align}
    Furthermore, the weighted logarithmic energy is given by $I_Q(\mu_Q) = \mathcal{I}_{ \rm s }(p,c,\tau)$, where
    \begin{align}
    \begin{split} \label{weighted energy simply connected}
      \mathcal{I}_{ \rm s }(p,c,\tau) &:= \frac{3}{4} + \frac{3c}{2} -\frac{cp^2}{1+\tau} +\frac{R^3\kappa p (2-3a^2-3\tau a^2 + 2\tau a^4) }{2(1-\tau^2)^2a^3}    \Big( 1-\tau  -\frac{ 2-3a^2+3\tau a^2-2\tau a^4}{  2-3a^2-3\tau a^2 + 2\tau a^4 } \frac{\kappa}{1-a^2}\Big)
    \\
    &\quad +2c(1+c)\log a + c^2 \log \Big( \frac{ c (1-\tau^2)(1-a^2) }{ R \kappa } \Big) - (1+c)^2 \log R. 
    \end{split}
    \end{align}
\end{itemize}
\end{thm}

\begin{figure}[t]
     \begin{subfigure}{0.32\textwidth}
        \begin{center}
            \includegraphics[width=\textwidth]{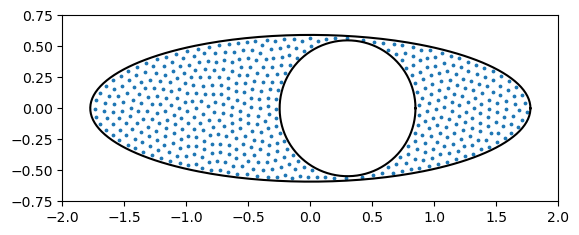}
        \end{center}
    \subcaption{ $p=0.3$ }
    \end{subfigure}
    \begin{subfigure}{0.32\textwidth}
        \begin{center}
            \includegraphics[width=\textwidth]{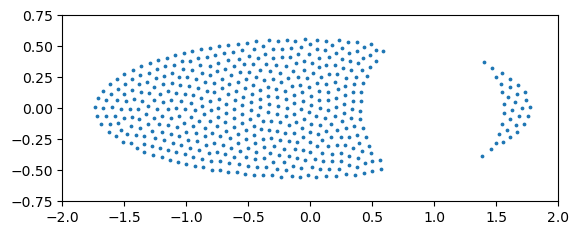}
        \end{center}
    \subcaption{ $p=1.0$}
    \end{subfigure}
     \begin{subfigure}{0.32\textwidth}
        \begin{center}
            \includegraphics[width=\textwidth]{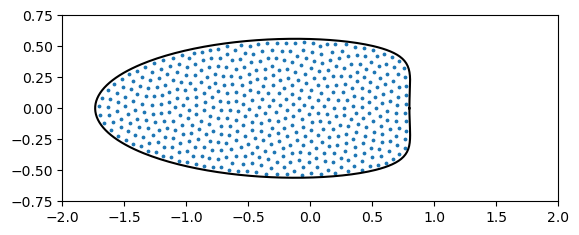}
        \end{center}
    \subcaption{  $p=1.5$ }
    \end{subfigure}
    \caption{The plots illustrate the configurations of Fekete points associated with the Hamiltonian \eqref{Ham complex}, with parameters \( c = 0.4 \), \( \tau = 0.5 \), and \( N = 500 \). Regime I corresponds to \( p < \sqrt{2/5} \approx 0.63 \), while Regime III approximately corresponds to \( p > 1.12 \). In (A) and (C), the black solid lines represent the boundaries of the droplets as determined in Theorem~\ref{Thm_droplet and energy}. } \label{Fig_Fekete points}
\end{figure}

We refer to Figure~\ref{Fig_Fekete points} for numerical verifications of the explicit shape of the droplet, cf. Remark~\ref{Rem_Fekete}. Additionally, the graphs of the energies  \eqref{weighted energy doubly connected} and \eqref{weighted energy simply connected} are presented in Figure~\ref{Fig_energies}.

As previously mentioned, Theorem~\ref{Thm_droplet and energy} on the description of the droplets extends the findings of \cite[Section 2]{BBLM15} for the $\tau = 0$ case, as well as those of \cite[Section 2.1]{By24} for the $p = 0$ case (see Remark~\ref{Rem_droplet extremal}).
Furthermore, Theorem~\ref{Thm_droplet and energy} on the evaluation of the energies generalises the results in \cite[Proposition 2.4]{BSY24} for the $\tau = 0$ case.
In both extremal cases, the doubly connected regime (Regime I) is referred to as the post-critical regime, while the simply connected regime (Regime II) for $\tau = 0$ or the two-component regime (Regime III) for $p = 0$ is referred to as the pre-critical regime.

We also note that Regime I in Definition~\ref{Def_regimes of p and c} corresponds to the case where, in the description of the droplet \eqref{droplet_doubly connected}, the outer ellipse does not intersect the inner circle. 
On the other hand, Proposition~\ref{prop_univalence} establishes that in Regime II, the rational map \eqref{eq_simply conformal map} is univalent and defines a conformal mapping from the exterior of the unit disc onto the exterior of the droplet.

\begin{figure}[t]
     \begin{subfigure}{0.32\textwidth}
        \begin{center}
            \includegraphics[width=\textwidth]{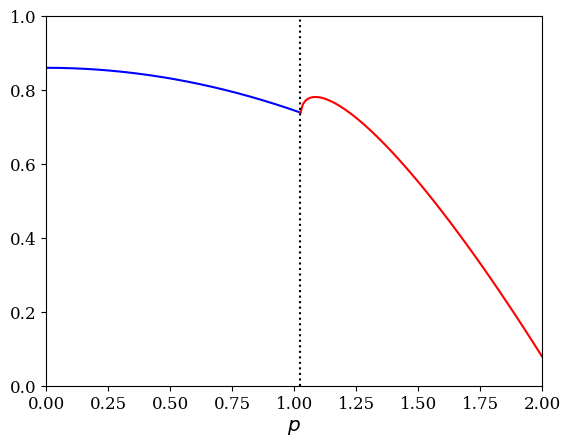}
        \end{center}
    \subcaption{ $c=0.15$ }
    \end{subfigure}
    \begin{subfigure}{0.32\textwidth}
        \begin{center}
            \includegraphics[width=\textwidth]{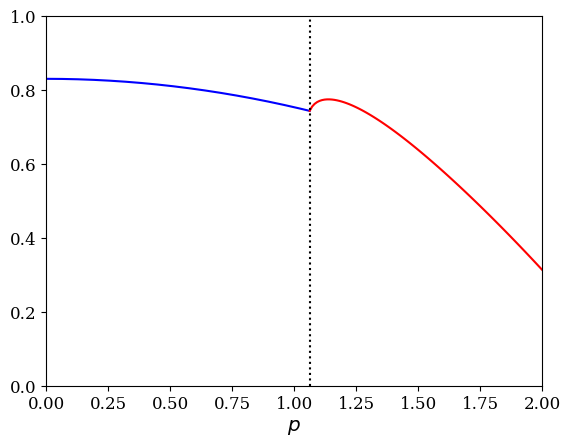}
        \end{center}
    \subcaption{ $c=0.1$}
    \end{subfigure}
     \begin{subfigure}{0.32\textwidth}
        \begin{center}
            \includegraphics[width=\textwidth]{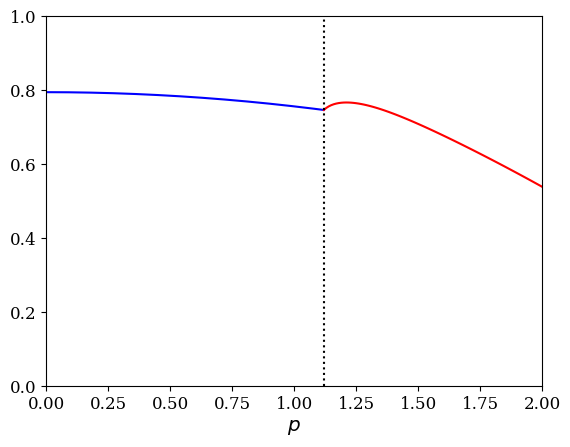}
        \end{center}
    \subcaption{  $c=0.05$ }
    \end{subfigure}
    \caption{ The plots display the graphs $p \mapsto I_Q(\mu_Q)$, given by  \eqref{weighted energy doubly connected} and \eqref{weighted energy simply connected}, for different values of $c$ with $\tau = 0.3$. The vertical dotted line indicates the critical value between Regimes I and II.  } \label{Fig_energies}
\end{figure}

\begin{rem}[Fekete points and numerics] \label{Rem_Fekete}
The discrete counterpart of the equilibrium measure is known as the Fekete point distribution, see e.g. \cite{Am17} and references therein. More precisely, we consider a configuration of points that minimises the Hamiltonians \eqref{Ham complex} and \eqref{Ham symplectic}. These configurations can be interpreted as the low-temperature (\( \beta = \infty \)) limit of the Coulomb gas ensembles. Since the macroscopic distribution of the Coulomb gas does not depend on the value of fixed \( \beta>0 \), the Fekete point configuration can be used to numerically observe the shape of the droplet. We also refer to \cite{AT24} for the $\beta$-ensembles with a flat equilibrium measure.
\end{rem}

\begin{rem}[The extremal $\tau=0$ case] \label{Rem_droplet extremal}
For the case $\tau = 0$, the rational map \eqref{eq_simply conformal map} simplifies, as the simple pole at the origin degenerates. Furthermore, the algebraic equations \eqref{eq:sc2}, \eqref{eq:sc3}, and \eqref{eq:sc4} can be solved more explicitly, leading to the expressions
\begin{equation} \label{R kappa tau0}
R|_{ \tau=0 }=\frac{1+p^2a^2}{2p a}, \qquad \kappa|_{ \tau=0 }=\frac{(1-a^2)(1-p^2a^2)}{1+p^2a^2}.  
\end{equation}
Here, $a$ satisfies $f(1/a) = p$ and $a^2=x$ is given as a unique solution to the cubic equation
  \begin{equation} \label{cubic eqn for tau0}
x^3 -\Big( \frac{p^2+4c+2}{2p^2} \Big) x^2+\frac{1}{2p^4} =0 
 \end{equation}
such that $0<a<1$ and $\kappa>0.$  Furthermore, as an immediate consequence of \eqref{weighted energy simply connected}, for the extremal case $\tau=0$, it follows that 
 \begin{align}
    \begin{split} \label{energy tau 0 simply}
      \mathcal{I}_{ \rm s }(p,c,\tau)\Big|_{\tau=0} &= \frac{3}{4} + \frac{3c}{2} -cp^2 +\frac{R^3\kappa p}{2 a^3}(2-3a^2) \frac{1-a^2-\kappa}{1-a^2}
      \\
    &\quad  +2c(1+c)\log a +c^2\log c-c^2\log \frac{R\kappa}{(1-a^2)}-(1+c)^2\log R .
    \end{split}
    \end{align}
Using \eqref{R kappa tau0} and \eqref{cubic eqn for tau0}, one can check that this formula is consistent with that derived in \cite[Proposition 2.1]{BSY24}. 
\end{rem}

\begin{rem}[Further phases in the general case] 
While some of our results, such as Theorem~\ref{Thm_droplet and energy} (i), can be naturally extended with minimal additional effort, our focus on the case \( p \in \mathbb{R} \) is primarily to make the phase diagram as explicit as possible. Another reason for this focus is that, when considering \eqref{Gibbs symplectic}, the potential must be symmetric with respect to the real axis. Consequently, for the symplectic ensemble, it is not meaningful to consider only a single point \( p \notin \mathbb{R} \).
On the other hand, if one considers more point charges at various points, then further phases can arise. The case with multiple point charges has also been studied in the literature including \cite{KKL24,LY23,BKP23,LY19}.
\end{rem}

\begin{rem}[Critical phases at intersections of different regimes] \label{Rem_critical phases}
There exist various critical regimes at intersections of different regimes. Let us summarise their geometric descriptions. 
\begin{itemize}
    \item The intersection of Regimes I and II is the case where, in \eqref{droplet_doubly connected}, the outer ellipse meets the inner circle tangentially at the rightmost edge, see Figure~\ref{Fig_critical phases} (A). 
    \smallskip 
    \item The intersection of Regimes I and III occurs when, in \eqref{droplet_doubly connected}, the outer ellipse meets the inner circle tangentially at two conjugate points in the upper and lower half-planes, see Figure~\ref{Fig_critical phases} (B).
    \smallskip 
    \item The intersection of Regimes II and III is less intuitive compared to the previous two cases. At criticality, this corresponds to the emergence of a new archipelago. In Hermitian random matrix theory, the analogous phenomenon has been studied under the name of the birth of a cut, see e.g. \cite{Ake97,BL08, Cl08, FK19}.
    \smallskip 
    \item The ``most'' critical case is the triple point where all three regimes intersect. For the generic case with $0< \tau <1 $ and $p > 0$, this critical point occurs when
\begin{align} \label{def of triple pts}
    c_{\rm tri}=  \frac{(1-\tau)^3}{2\tau(3+\tau^2)}, \qquad p_{\rm tri} = 2\sqrt{\frac{2\tau(1+\tau)}{3+\tau^2}}. 
\end{align} 
    In this case, the outer ellipse again meets the inner circle tangentially at the rightmost edge. Additionally, the curvatures of the ellipse and the circle at this point are identical, see Figure~\ref{Fig_critical phases} (C).
\end{itemize}

From the perspective of the ensembles \eqref{Gibbs complex} and \eqref{Gibbs symplectic}, several interesting features emerge at such criticality. For the complex Ginibre ensemble, where the critical regime occurs at the intersection of Regimes I and II, the local statistics were recently explored in \cite{KLY23}. In this context, the Painlevé II critical asymptotics arise. Such emergence of critical behaviour is consistent with findings in Hermitian random matrix theory at multi-criticality \cite{BI03,CK06,CKV08}, where the global density vanishes at a bulk point with quadratic decay. Furthermore, a recent study \cite{BSY24} demonstrated that, at this critical point, the Tracy-Widom distribution appears in the constant term of the free energy expansion. This contrasts with the regular case, where the zeta-regularised determinant of the Laplacian is believed to arise (cf. \cite{ZW06}).

Such problems in our present model, the elliptic Ginibre ensembles with a point charge, remain widely open. We expect that the local statistics at the critical points arising at the intersections of Regimes I and II, as well as Regimes I and III, correspond to Painlevé II critical asymptotics, thereby being contained in the same universality class introduced in \cite{KLY23}.  
Similarly, we expect the Tracy-Widom distribution to emerge in the free energy expansions.
In addition, perhaps the most intriguing case in our model is the triple point. At this level of criticality, one might expect the asymptotic behaviours to exhibit the critical asymptotics of higher-order multi-criticality. More precisely, it can be expected that the critical asymptotic behaviour of the Hermitian random matrix model, whose global density vanishes at a bulk point with higher (quartic) order decay, may arise. Consequently, one may expect the Painlevé II hierarchy to emerge in this case. 
\end{rem}

\begin{figure}[t]
     \begin{subfigure}{0.32\textwidth}
        \begin{center}
            \includegraphics[width=\textwidth]{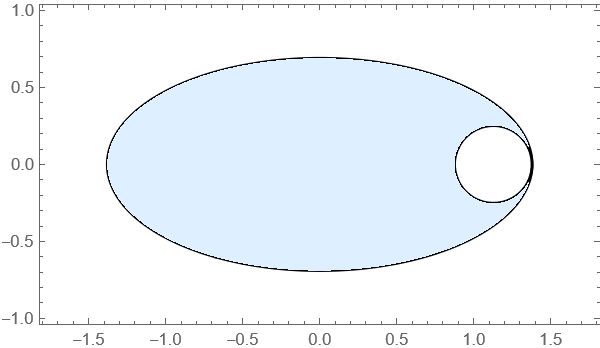}
        \end{center}
    \subcaption{ Regimes I \& II }
    \end{subfigure}
  \begin{subfigure}{0.32\textwidth}
        \begin{center}
            \includegraphics[width=\textwidth]{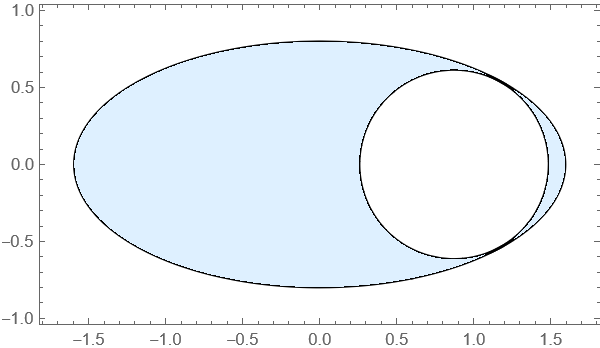}
        \end{center}
    \subcaption{ Regimes I \& III }
    \end{subfigure}
     \begin{subfigure}{0.32\textwidth}
        \begin{center}
            \includegraphics[width=\textwidth]{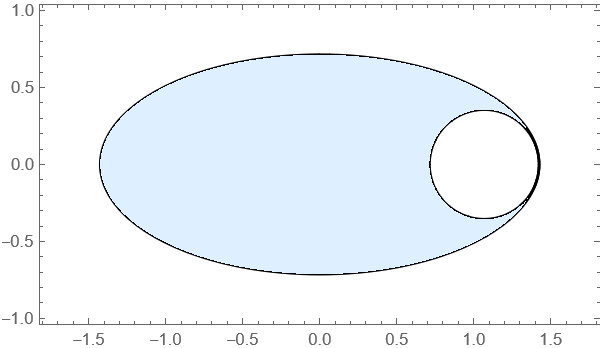}
        \end{center}
    \subcaption{  Regimes I, II \& III }
    \end{subfigure}
    \caption{The plots illustrate various critical phases with \(\tau = 1/3\). In this case, the triple points \eqref{def of triple pts} are given by \(c_{\rm tri} = 1/7\) and \(p_{\rm tri} =  2\sqrt{14}/7 \approx 1.07  \). Plot (A) represents the case \(c = 1/14 < c_{\rm tri}\), where the intersection of Regimes I and II occurs at \(p = 2\sqrt{7}(\sqrt{30} - 1)/21 \approx 1.13 \). Plot (B) corresponds to the case \(c = 3/7 > c_{\rm tri}\), where \(p = 4/\sqrt{21} \approx 0.88\). Plot (C) shows the case when \(c = c_{\rm tri}\) and \(p = p_{\rm tri}\), where the droplet exhibits identical curvatures for the ellipse and circle at the singularity. } \label{Fig_critical phases}
\end{figure}

\begin{rem}[Multi-component ensembles in Regime III]
Compared to the Ginibre case when $\tau=0$, one of the interesting features arising in the elliptic case is the emergence of the multi-component ensemble in Regime III. As previously mentioned, the multi-component ensemble has recently gained significant attention, as it exhibits non-trivial additional statistical properties in both the fluctuations \cite{ACC23a,AC24,ACC23}, represented by the Heine distribution, and various free energy expansions \cite{ACC23,ACCL24, Ch22, Ch23}, which involve oscillatory asymptotics expressed in terms of the Jacobi theta function.
The analogous setup in Hermitian random matrix theory is called the multi-cut regime, where several results are known, see e.g. \cite{CFWW21,BG24} and references therein.

We also mention that for the case $p = 0$, the symmetry of the potential \eqref{eq:potential} with respect to the origin makes it possible to derive the precise shape of the droplets in Regime III. The key idea here is to remove the symmetry, thereby reformulating the problem into an equivalent one where the associated droplet is simply connected. Further details can be found in \cite[Remark 1.8]{By24}.  
\end{rem}

\begin{rem}[Phases of the motherbody]
In the study of the point processes \eqref{Gibbs complex} and \eqref{Gibbs symplectic}, a natural object of interest is the planar orthogonal polynomials \( P_j \) associated with the weight \( e^{-N W(z)} \). The limiting zero distribution of \( P_j \) as the degree increases is called the motherbody (or the potential-theoretic skeleton), which is typically a one-dimensional subset of the droplet. The motherbody plays a key role in the asymptotic behaviour of orthogonal polynomials. Clearly, the topology of the motherbody depends strongly on that of the convex hull of the droplet.  
In addition, the motherbody often exhibits more diverse topological phases. Specifically, within the same topological type of the droplet, further phases of the motherbody may arise. This phenomenon is closely related to the number of critical points of the function $z \mapsto \int_\mathbb{C} \log\frac{1}{|z-w|^2} \, \mathrm{d}A(w) + W(z)$,
which naturally appears in the variational condition \eqref{eq:variational}.  
In our present setting, we expect that in Regime II, where the droplet is simply connected, two distinct phases of the motherbody exist, depending on the number of critical values (cf. Lemma~\ref{lem:critical points}). For recent developments in this direction, we refer to \cite{KKL24,BS20} and references therein.  
\end{rem}

\begin{rem}[Equilibrium measure in the Hermitian limit $\tau=1$]\label{rem_extermal tau=1}
The elliptic Ginibre ensembles provide a natural bridge between Hermitian and non-Hermitian random matrix theories \cite{FKS97,FKS97a,FKS98}. This characteristic becomes particularly evident in the Hermitian limit $\tau \uparrow 1$, where various intriguing regimes emerge, see, e.g. \cite{ACV18,ADM24,AB23,BES23,Kan02} and references therein for studies on the complex and symplectic elliptic Ginibre ensembles in the almost-Hermitian regime.
 
We briefly discuss the Hermitian limit of the two-dimensional equilibrium measure. First, by taking the limit $\tau \uparrow 1$ of the potential \eqref{eq:potential}, we obtain  
\begin{equation} 
\lim_{ \tau \uparrow 1 } Q(x+iy)=V(x):=\begin{cases}
\displaystyle \frac{x^2}{2}-2c \log|x-p|, &\text{if }y=0,
\smallskip 
\\
+\infty &\text{otherwise}.
\end{cases} 
\end{equation} 
It follows from the standard methods of logarithmic potential theory
\cite{ST97} that the associated one-dimensional equilibrium measure
$\mu_V$ is given by 
\begin{equation}
    \frac{\ud \mu_V(x)}{\ud x} = \frac{1}{\pi}\sqrt{-R(x)} \cdot \mathbbm{1}_{R(x)\le 0}, \qquad  R(z) = \bigg(\int_\R \frac{\ud\mu_V(s)}{z-s}-\frac{V'(z)}{2}\bigg)^2=\frac{z^2}{4}+ \frac{Az^2 +Bz +C}{(z-p)^2},
\end{equation}
where the constants $A$, $B$, and $C$ in the rational function $R(z)$
remain to be determined.
As a consequence, the equilibrium measure $\mu_V$ is supported on either a single interval or two disjoint intervals. 
The explicit characterisation of $\mu_V$ is also standard and is analogous to that in the two-dimensional case. One first assumes that $\mu_V$ is one-cut regular and then verifies the Euler--Lagrange conditions.\footnote{We note that the characterisation of the equilibrium measure given as a remark in \cite[Eq.~(2.10)]{By24} contains is incorrect. Specifically, the equilibrium measure is assumed to be supported on two intervals for all parameter values.}
The transition between the one-cut and two-cut regular phases is
characterised by the birth of a cut, which is consistent with our
observation regarding the interface between Regimes II and III in
Remark~\ref{Rem_critical phases}. We also refer to
\cite{MOR15,OS18} for studies of the birth of a cut phenomenon in the
one-dimensional setting from the perspective of equilibrium measures.

Furthermore, the three regimes for $\tau \in (0,1)$ extend continuously to the Hermitian limit $\tau \uparrow 1$. More precisely, for $c>0$, Regime~I
degenerates to the empty set, Regime~II converges to the one-cut
regular phase, and Regime~III converges to the two-cut regular phase; see Figure~\ref{fig:phase}.   
In this limit, the triple point~\eqref{def of triple pts} converges to $c=0, p=2$, which agrees with the absence of a third phase in the one-dimensional case.
\end{rem}

\begin{rem}[Continuity between Regimes I and II]\label{rem_doubly simply intersection}
It is intuitively clear that there is a smooth transition between different regimes. We shall discuss this aspect via explicit computations between Regimes I and II. For this purpose we take $a \to 1$ limit in Regime II, see \cite[Remark 2.6]{BFL25} for a similar discussion.  
To be more precise, let $\epsilon=\kappa/(1-a^2)$ be fixed. Then as $ a \to 1$, we have 
\begin{align*}
    c\to \frac{\epsilon^2}{1-\tau^2-\epsilon^2},\qquad
    p \to (1+\tau-\epsilon)\Big(\frac{1-\tau^2}{1-\tau^2-\epsilon^2}\Big)^{1/2}, 
    \quad 0\leq \epsilon\leq \frac{(1-\tau)^2}{1+\tau},
\end{align*}
which coincides with the parameterisation 
\begin{align}\label{eq_boundary of Regime I-1}
    p = (1+\tau)\sqrt{1+c}-\sqrt{c(1-\tau^2)}, \qquad 0\leq c\leq c_{\rm tri} = \frac{(1-\tau)^3}{2\tau(3+\tau^2)}
\end{align}
of the boundary of Regime I.  In addition to the continuity of the parameter space, one can also observe the continuity of the energies \eqref{weighted energy doubly connected} and \eqref{weighted energy simply connected}.  
To see this, note that from~\eqref{eq:sc2} and~\eqref{eq:sc3}, we have  
\begin{align*}
    R^2\epsilon^2 \to c(1-\tau^2), \qquad R^2 \to 1+c, \qquad \kappa\to 0,
\end{align*}  
as \( a \to 1 \). Therefore, it follows that \( \mathcal{I}_{\rm s} \to \mathcal{I}_{\rm d} \) as \( a \to 1 \), see Figure~\ref{Fig_energies}.   
\end{rem}

We now turn our focus to a more application-oriented perspective on elliptic Ginibre matrices. As previously mentioned, the insertion of a point charge has an equivalent formulation in terms of the moments of characteristic polynomials. For the Ginibre ensembles with $\tau=0$, the moments of the characteristic polynomial $\mathbb{E} |\det (G-z)|^{\gamma}$ have been studied both for fixed $\gamma=O(1)$ \cite{WW19,DS22} and in the exponentially varying regime $\gamma = O(N)$ \cite{BSY24}. These studies have found various applications in areas such as Gaussian multiplicative chaos. (See also \cite{FG23} for a recent study on the convergence of the characteristic polynomial of the complex elliptic Ginibre matrix.) 

\begin{cor}[\textbf{Moments of the characteristic polynomials of the elliptic Ginibre matrices}] \label{Cor_moments}
Let $c >0$ and $\tau \in [0,1)$ be fixed. Let $z \in \mathbb{R}$. Then as $N \to \infty$, we have 
\begin{equation} \label{Moments asymp}
\log \mathbb{E} \Big[ \,  \Big|\det (X-z) \Big|^{2cN} \Big]  = 
\begin{cases}
    \KK N^2 +\EE_N, &\textup{for the complex case} \smallskip
    \\
    2\,\KK N^2 +\EE_N, &\textup{for the symplectic case} 
\end{cases}
\end{equation}
where $\mathcal{K}$ is given as follows.
\begin{itemize}
    \item[\textup{(i)}] If $(p=z,c, \tau)$ falls within \textup{Regime I}, 
      \begin{equation}
    \mathcal{K} = \frac{c \, z^2}{1+\tau}  -\frac{3c}{2}  +\frac{(1+c)^2}{2}\log(1+c)  - \frac{c^2}{2}\log\Big( c(1-\tau^2) \Big)  . 
    \end{equation}
      \item[\textup{(ii)}] If $(p=z,c, \tau)$ falls within \textup{Regime II}, 
      \begin{align}
    \begin{split}  
   \mathcal{K}&= \frac{c\,z^2}{1+\tau}  - \frac{3c}{2} - \frac{R^3\kappa  (2-3a^2-3\tau a^2 + 2\tau a^4) }{2(1-\tau^2)^2a^3}    \Big( 1-\tau  -\frac{ 2-3a^2+3\tau a^2-2\tau a^4}{  2-3a^2-3\tau a^2 + 2\tau a^4 } \frac{\kappa}{1-a^2}\Big) z 
    \\
    &\quad -2c(1+c)\log a - c^2 \log \frac{ c (1-\tau^2)(1-a^2) }{ R \kappa } + (1+c)^2 \log R,
    \end{split}
    \end{align}  
    where $(R,a, \kappa)$ is a solution to the coupled algebraic equations \eqref{eq:sc2}, \eqref{eq:sc3} and \eqref{eq:sc4} with $p=z$. 
\end{itemize} 
Here, the error term is given by 
\begin{equation} \label{error term in cor}
\mathcal{E}_N = \begin{cases}
 o(N^{1/2+\epsilon}) &\textup{for the complex case},
\smallskip 
\\
 o(N^2)  &\textup{for the symplectic case},
\end{cases}
\end{equation}
for some $\epsilon>0.$
\end{cor}

Note that in terms of the weighted logarithmic energies in Theorem~\ref{Thm_droplet and energy}, we have 
\begin{equation}
\mathcal{K}= -I_Q(\mu_Q) + \frac{3}{4}, 
\end{equation}
where $p$ is identified with $z$. 
One also notices that the error term in \eqref{error term in cor} becomes significantly more precise in the complex case compared to its symplectic counterparts. This is because, for the complex case, we can utilise recent progress on general Coulomb gases \cite{LS17,BBNY19,Se23,AS21}, whose symplectic (or Neumann) counterparts are not yet available in the literature. Nonetheless, there are some general conjectures that we can use to formulate a conjecture in our present case. Let us introduce these in a separate remark. 

\begin{rem}[Conjecture on precise error terms] Due to the general conjecture presented in several works, such as \cite{CFTW15, ZW06, JMP94, TF99} for the complex case and \cite{BKS23} for the symplectic case, we expect that the optimal error term takes the following form.
\begin{itemize}
    \item For the complex case, 
\begin{equation}
\mathcal{E}_N= \begin{cases}
\frac{1}{12} \log N +O(1) &\textup{in Regime I},
\smallskip 
\\
O(1) &\textup{in Regime II}. 
\end{cases} 
\end{equation}
\item For the symplectic case,
\begin{equation}
\mathcal{E}_N=  \bigg(  \int_{ \mathsf{E} } \log |z-\bar{z}|  \,\ud A(z) -  \int_{ S } \log |z-\bar{z}|  \,\ud A(z) \bigg) N +  \begin{cases}
\frac{1}{24} \log N +O(1) &\textup{in Regime I},
\smallskip 
\\
O(1) &\textup{in Regime II}. 
\end{cases} 
\end{equation}
Here, $\mathsf{E}$ is given by \eqref{def of elliptic law}. 
\end{itemize}
These precise asymptotic expansions, particularly the coefficients in the $O(\log N)$ terms, reveal a characteristic dependence on the topological properties of the droplet. 
To be more precise, it is believed that the coefficients of the $\log N$ term in the asymptotic expansion of the free energy are given by 
\begin{equation}
\frac12- \frac{\chi}{12}, \qquad \frac12- \frac{\chi}{24},
\end{equation}
for the complex and symplectic cases, respectively, where $\chi$ is the Euler characteristic of the droplet. 
Moreover, the $O(1)$ terms are intriguing as well, as they are believed to be intricately linked to certain conformal geometric structures associated with the equilibrium measures.
\end{rem}

\begin{rem}[Duality relation] For various random matrix ensembles, the moments of characteristic polynomials exhibit duality relations. This implies an identity whereby the averages of characteristic polynomials can be expressed in terms of the averages of certain observables over other random matrices, with the dimension of the latter determined by the exponent of the moments. For a more detailed discussion of this topic, we refer to the recent review \cite{Fo25}. This duality has significant applications in asymptotic analysis, particularly establishing connections between the electrostatic energies of different ensembles, see e.g. \cite{BSY24,BFL25}. In the case of elliptic Ginibre matrices, the duality relation is derived using Grassmann integration methods, see \cite[Proposition 5.1]{Sere23} and \cite[Remark 5.1]{Fo25}. 
\end{rem}

\begin{rem}[Real elliptic Ginibre matrices] In this work, we focus on the characteristic polynomial of complex and symplectic elliptic Ginibre matrices. On the other hand, the third symmetry class, the real elliptic Ginibre matrices, is also of significant interest and has been actively studied, along with various applications, see e.g. \cite[Section 7.9]{BF24}. A key distinction is that, in the real case, due to the nontrivial probability of having purely real eigenvalues, the eigenvalue distribution no longer forms a one-component plasma system like \eqref{Gibbs complex} and \eqref{Gibbs symplectic} but instead a two-species system. Nonetheless, since the number of real eigenvalues in the strongly non-Hermitian regime (where $\tau<1$ is fixed) is of subdominant order $O(\sqrt{N})$ \cite{EKS94,FN08,BKLL23}, the macroscopic behaviour (such as the elliptic law) of real elliptic Ginibre matrices remains largely consistent with that of their complex and symplectic counterparts. This, in turn, should also apply to the limiting spectral distribution of conditional elliptic ensembles and the moments of characteristic polynomials. In particular, we expect the same leading-order term, $\mathcal{K} N^2$, in \eqref{Moments asymp} to appear in the case of real elliptic Ginibre matrices. This gives rise to the exponentially varying counterpart of recent findings \cite{Kiv24,SS24,Fyo16,FT21} for a fixed exponent.
\end{rem}

\subsection*{Organisation of the paper} In Section~\ref{Section_topology}, we employ the theory of quadrature domains developed in \cite{LM16} to establish the first part of Theorem~\ref{Thm_main droplet}. Section~\ref{Section_doubly} focuses on the characterisation of the doubly connected domain and the evaluation of its associated energy. Section~\ref{Section_simply} proceeds in parallel, focusing on the simply connected domain. The final section brings together the results from the preceding sections to conclude the proofs of the main results, Theorems~\ref{Thm_main droplet} and~\ref{Thm_droplet and energy}. Furthermore, Corollary~\ref{Cor_moments} is addressed in Subsection~\ref{Subsection_moments}.

\subsection*{Acknowledgements} Sung-Soo Byun was supported by the National Research Foundation of Korea grant (RS-2025-00516909, RS-2026-25518141). The authors express their gratitude to Peter Forrester, Arno Kuijlaars, Sampad Lahiry, Kohei Noda, Meng Yang and Lun Zhang for their interest and helpful discussions.

\section{Quadrature domain and topological characterisation of the droplet} \label{Section_topology}

In this section, we recall some known facts about quadrature domains associated with the logarithmic potential, primarily from the work \cite{LM16} of Lee and Makarov. By employing this theory, we prove the first part of Theorem~\ref{Thm_main droplet}, namely, the assertion that the droplet is either doubly connected, simply connected, or composed of two disjoint simply connected components. For a general introduction to quadrature domain theory and related fields, we refer to~\cite{GS05,LM16} and references therein.

Recall that a bounded connected open set $\Omega \subset \mathbb{C}$ is called a \emph{bounded quadrature domain} if it satisfies a \emph{quadrature identity}: for all integrable analytic functions $f$ on $\Omega$, 
\begin{align}\label{eq_quad identity}
    \int_\Omega f(\zeta) \ud A(\zeta)= \sum_{k=1}^n c_k f^{(m_k)}(a_k),
\end{align}
where $n \in \mathbb{Z}_{\geq 0}$ and $(c_k, m_k, a_k)_{k=1}^n \subset \mathbb{C}\setminus \{0\}\times \mathbb{Z}_{\geq 0}\times \Omega$. 
Here, the points $a_k$ are not necessarily distinct. For simplicity, we assume that $\partial \Omega$ is a Jordan curve. The \emph{quadrature function} $r_\Omega$ of a quadrature domain $\Omega$ is defined as 
\begin{align} \label{def of quad function}
    r_\Omega(\zeta) := \sum_{k=1}^{n} c_k\frac{m_k!}{(\zeta-a_k)^{m_k+1}}.
\end{align}
Then the quadrature identity \eqref{eq_quad identity} can be rewritten as 
\begin{align} \label{eq_quad identity v2}
    \int_\Omega f(\zeta) \ud A(\zeta) = \frac{1}{2\pi i}\int_{\partial \Omega} f(\zeta)r_\Omega(\zeta) \ud \zeta.
\end{align}
We define the degree of $r_\Omega$ as the \emph{order} of a quadrature domain $\Omega$.
Here the degree of a rational function is defined as the maximum of the degree of denominator and numerator in reduced form of the rational function.

The above definition can naturally be extended to unbounded quadrature domains. For this, let $\Omega \subset \hat{\mathbb{C}}$ be a connected open set where $\hat{\mathbb{C}}=\mathbb{C} \cup \{ \infty \}$ is the extended complex plane. Assume that $\infty \in \Omega$ and $\partial \Omega$ is a Jordan curve. Then $\Omega$ is an \emph{unbounded quadrature domain} if it satisfies a quadrature identity~\eqref{eq_quad identity} for all  integrable and analytic functions $f$ on $\Omega$ such that $f(\infty)=0$.  
Observe that if an unbounded quadrature domain \(\Omega\) of order \(d \geq 0\) does not contain the origin in its closure, inversion with respect to the unit circle conformally maps \(\Omega\) to a bounded quadrature domain of order \(d+1\). Conversely, a bounded quadrature domain containing the origin can be conformally transformed into an unbounded quadrature domain via circular inversion.

\begin{rem}[Boundary of a simply connected quadrature domain]
\label{rem_boundary of simply connected QD}
A well-known result due to Aharonov and Shapiro~\cite{AS76} states that a simply connected domain \(\Omega\) is a quadrature domain if and only if there exists a rational univalent function \( f \) with specific conformal mapping properties:
\begin{itemize}
    \item if \(\Omega\) is a bounded domain, then \( f \) conformally maps the unit disc \(\mathbb{D}\) onto \(\Omega\), with all its poles lying outside the closed unit disc \(\bar{\mathbb{D}}\);
    \smallskip
    \item if \(\Omega\) is unbounded, then \( f \) conformally maps the complement of the closed unit disc, \(\bar{\mathbb{D}}^c\), onto \(\Omega\), with all its poles lying inside \(\mathbb{D}\) except for a single simple pole.
\end{itemize}
Once the quadrature function of a simply connected quadrature domain is given, one can derive explicit rational univalent functions that describe the boundary of the quadrature domain via the so-called conformal mapping method.   
\end{rem}

Let us provide some examples of bounded and unbounded quadrature domains, which are closely related to our model of interest.  
\begin{itemize}
    \item[(a)] The open disc \(\Omega = \mathbb{D}(p, \rho)\) is a bounded quadrature domain of order 1 with the quadrature function \(r_\Omega(\zeta) = \rho^2/(\zeta - p)\). Indeed, using the mean value property and the Cauchy integral formula, one can observe that
    \begin{align*}
        \int_{\mathbb{D}(p,\rho)} f(\zeta) \ud A(\zeta) = \rho^2 f(p) = \frac{1}{2\pi i }\int_{|\zeta-p|=\rho} f(\zeta)r_\Omega(\zeta)\ud \zeta.
    \end{align*}
    \item[(b)] The exterior of a closed disc, \(\Omega = \bar{\mathbb{D}}(p, \rho)^c\), can similarly be shown to be an unbounded quadrature domain of order 0, with the quadrature function \(r_\Omega(\zeta) = \bar{p}\).
    \smallskip 
    \item[(c)] An example of an unbounded quadrature domain of order 1 is the exterior of an ellipse with major and minor axes given by \(1 \pm \tau\), where \(\tau \in [0, 1)\). The associated conformal map is the Joukowsky transform \(f(z) = z +  \tau/z \). By virtue of Remark~\ref{rem_boundary of simply connected QD}, \(f\) is a univalent rational function defined on \(\bar{\mathbb{D}}^c\) with two simple poles at \(0\) and \(\infty\), one located inside the open unit disc and the other outside the closed unit disc.
\end{itemize}

\begin{rem}[Uniqueness problem of quadrature domain]\label{rem_uniqueness of QD}
It is known that if two bounded quadrature domains of order \(\leq 2\) share the same quadrature function, they are identical, see \cite[Theorem 10 and Corollary 10.1]{Gusta83}. In particular, order 1 bounded quadrature domains are precisely the open discs described above. 

In contrast, the unbounded analogue of this result is more subtle. For instance, infinitely many order-zero unbounded quadrature domains of the form \(\bar{\mathbb{D}}(p, \rho)^c\) share the same quadrature function \(\bar{p}\). 
Nevertheless, the uniqueness of unbounded quadrature domains emerges under additional constraints on the area of the domain's complement. More precisely, it was shown in \cite[Theorem 2.1]{LM16} that if two unbounded quadrature domains have complements of equal area and share the same quadrature function, which is a polynomial of degree at most \(2\), then they must be identical.
\end{rem}

We now illustrate the connection between quadrature domains and logarithmic potential theory. For this purpose, let us first recall that for a Borel measurable set $\Omega \subset \mathbb{C}$ with compact boundary, the \emph{Cauchy transform} $C_\Omega$ is defined by 
\begin{align}\label{eq_Cauchy transform}
    C_\Omega(\zeta) := \int_\Omega \frac{1}{\zeta-\eta}\ud A(\eta).
\end{align}
Next, assume that $\Omega \subsetneq \hat{\mathbb{C}}$ and $\partial \Omega$ is a Jordan curve such that $\infty \notin \partial \Omega$. If a continuous function $\mathsf{S}_\Omega: \bar{\Omega}\to \hat{\mathbb{C}}$ is meromorphic on $\Omega$ and satisfies
\begin{align} \label{def of Schwarz ftn}
    \mathsf{S}_\Omega (\zeta) = \bar{\zeta}, \qquad \zeta \in \partial \Omega,
\end{align}
we call the function $\mathsf{S} \equiv \mathsf{S}_\Omega $ the \emph{(one-sided) Schwarz function} of \(\Omega\). 
It is clear that if a Schwarz function exists for a domain $\Omega$, then it is unique.
The Schwarz function plays a significant role in the theory of quadrature domains. In particular it is well known that \(\Omega\) is a quadrature domain if and only if it has a Schwarz function. Moreover, in this case, the following holds:  
\begin{equation}\label{Schwarz-QD}
    \mathsf{S}_\Omega(\zeta) = r_\Omega(\zeta) + C_{\Omega^c}(\zeta), \qquad \zeta \in \bar{\Omega},
\end{equation}  
where $r_\Omega$ is the quadrature function \eqref{def of quad function}. 
For a simply connected quadrature domain $\Omega$, we have
\begin{equation}
\mathsf{S}_\Omega (\zeta) =  f(1/F(\zeta)), 
\end{equation}
where $f$ is a rational univalent map associated with $\Omega$ in Remark \ref{rem_boundary of simply connected QD}, and $F$ is its conformal inverse.
For more details, we refer to~\cite[Lemma 3.1]{LM16} and references therein. 

We now discuss the topological properties of quadrature domains within the framework developed in ~\cite{LM16}.
By definition, a Hele-Shaw-type potential $W$ (the potential with constant $\Delta W$) is called an \emph{algebraic potential} if it is a real valued function defined on an open subset of $\mathbb{C}$, which takes the form 
\begin{align}
    \frac{1}{t}W(\zeta) = |\zeta|^2 - H (\zeta), \qquad t>0,  
\end{align}
where $h:=\partial H$ is a rational function in the variable $\zeta$. 

\begin{thm}[cf. Theorems 2.3 and 3.3 in \cite{LM16}] \label{thm_S^c=QD}
Let $S \equiv S_W$ be the droplet associated with an algebraic potential of the form $W(\zeta)= |\zeta|^2 -H(\zeta)$. Then $S^c$ is a finite union of disjoint
quadrature domains $\Omega_1, \ldots, \Omega_q$, and their quadrature functions $r_1, \ldots, r_q$ satisfy
\begin{equation} \label{LM sum of quad functions}
    r_1(\zeta) + \ldots +r_q(\zeta) = h(\zeta).
\end{equation} 

Let $d$ be the degree of the rational function $h$. Assume that the boundary of $S$ is smooth and consists of disjoint Jordan curves, called the ovals. Let $q_j$ be the number of quadrature domains as components of $\hat{\mathbb{C}}\setminus S$ with connectivity $j \geq 1$. Then we have 
\begin{equation}\label{eq_connectivity bound}
    \#(\text{ovals}) + q_\text{odd} + 4(q-q_1) \leq 2d+2.
\end{equation}
Here, $q = \sum q_j$ and $q_\text{odd} = \sum q_{2j+1}$.
\end{thm}

We also mention that the connectivity bound is indeed sharp \cite{LM14}. 

\begin{rem} [Boundary of a general quadrature domain] 
In general, a multiply connected quadrature domain has an ``almost'' algebraic boundary, meaning that for any quadrature domain $\Omega$ of order $d$, there exists a polynomial $P(\zeta, \eta)$ such that
\begin{align*}
    \partial \Omega \subset \{\zeta\in \mathbb{C}:P(\zeta,\bar{\zeta})=0\}. 
\end{align*}
Here, two sets differ by only a finite number of points called \emph{special points}, see \cite[Section 12]{GS05}.  
The polynomial $P(\zeta, \eta)$ is symmetric with respect to $\zeta$ and $\eta$, and is of degree $d$ in each variable if $\Omega$ is a bounded quadrature domain, or degree $d+1$ if it is an unbounded quadrature domain.  
The Schottky double of $\Omega$, a closed Riemann surface constructed by gluing the boundaries of two copies of $\Omega$, can be identified as a real algebraic curve given by $P(\zeta, \mathsf{S}(\zeta)) = 0$.  
From the perspective of logarithmic potential theory, the algebraic equation $P(\zeta, \mathsf{S}(\zeta)) = 0$ is often referred to as the \emph{spectral curve}, which appears in various contexts, see e.g. \cite{BK12, BS20, KKL24}. 
\end{rem}

We are now ready to prove the first assertion in Theorem~\ref{Thm_main droplet}.

\begin{proof}[Proof of the first assertion in Theorem~\ref{Thm_main droplet}]

Recall that the potential $Q$ is given by \eqref{eq:potential}. 
Let us define 
\begin{equation}  \label{eq:quad h}
\mathsf{H}(\zeta):= - \tau \re \zeta^2 - 2c(1-\tau^2)\log |\zeta-p|, \qquad \mathsf{h} (\zeta) := \partial \mathsf{H}(\zeta) =\tau \zeta + \frac{c(1-\tau^2)}{\zeta-p}. 
\end{equation}
The potential~\eqref{eq:potential} is clearly an algebraic potential since
\begin{equation}
 (1-\tau^2)Q(\zeta)  = |\zeta|^2 - \mathsf{H}(\zeta). 
\end{equation}  
Note that the quadrature function $\mathsf{h}$ in \eqref{eq:quad h} is of degree 2. 
By Sakai’s laminarity theorem \cite{Sa10,Sa92}, we can exclude the critical cases and assume that the boundary $\partial S$ is smooth and consists of finitely many Jordan curves, see also \cite[Section 5.1]{LM16}.
Then in our present case, the connectivity bound \eqref{eq_connectivity bound} reads as
\begin{equation}
    \#(\text{ovals}) + q_\text{odd} + 4(q-q_1) \leq 6.
\end{equation}
Therefore, the possible configurations are as follows:
\begin{enumerate}[label=(\alph*)]
    \item $\#\text{(ovals})=1$, $q_1=1$, $q_j=0$ for $j\geq 2$;
    \smallskip 
    \item $\#\text{(ovals})=2$, $q_1=2$, $q_j=0$ for $j\geq 2$;
     \smallskip 
    \item $\#\text{(ovals})=2$, $q_1=0$, $q_2 =1$, $q_j=0$ for $j\geq 3$;
     \smallskip 
    \item $\#\text{(ovals)}=3$, $q_1=3$, $q_j=0$ for $j\geq 2$.
\end{enumerate}
It is evident that (a), (b), (c), and (d) correspond to simply connected droplet, doubly connected droplet, droplet composed of two simply connected components, and triply connected droplet, respectively. These configurations are also presented in  \cite[Examples $d=0,1,2$ below Theorem 2.3]{LM16}.

We claim that in fact, the droplet $S$ cannot be triply connected. Suppose that the droplet $S$ is triply connected. Then $S^c$ consists of an unbounded component $\Omega_1$ and bounded components $\Omega_2, \Omega_3$, which are all simply connected quadrature domains. Denote by $r_1, r_2$ and $r_3$ the quadrature functions of $\Omega_1, \Omega_2$ and $\Omega_3$, respectively. Then it follows from \eqref{LM sum of quad functions} and \eqref{eq:quad h} that
\begin{align} \label{quad identity for triply}
    r_1(\zeta)+ r_2(\zeta) + r_3(\zeta) =  \tau \zeta + \frac{c(1-\tau^2)}{\zeta-p}.
\end{align}
Observe that since $\Omega_2$ and $\Omega_3$ are bounded quadrature domains of order $\geq 1$, both $r_2$ and $r_3$ must contain distinct poles located in $\Omega_2$ and $\Omega_3$, respectively. In contrast, $r_1$ does not contain a pole in $\Omega_2 \cup \Omega_3$. Therefore, $r_1 + r_2 + r_3$ must have at least two poles in $\mathbb{C}$. This leads to a contradiction, as the right-hand side of \eqref{quad identity for triply} contains only a single pole at $p$ in $\mathbb{C}$. Hence, we conclude that the droplet $S$ can only be doubly connected, simply connected, or composed of two disjoint connected components.
\end{proof}

\begin{rem}
It is clear from the proof that for $\tau > 0$, the connectivity bound \eqref{eq_connectivity bound} does not fully characterise the possible topological types of the droplet, since a triply connected domain may emerge from the connectivity bound, but it does not genuinely occur in practice. In contrast, in the extremal case $\tau = 0$, the connectivity bound is sufficient to completely determine the topological types of the droplet. More precisely, when $\tau = 0$, the quadrature function has degree 1. Consequently, the possible topological types of the droplet deduced from \eqref{eq_connectivity bound} are either simply connected or doubly connected, which is indeed the case, as discussed in Remark~\ref{Rem_droplet extremal}.
\end{rem}

\section{Doubly connected droplet} \label{Section_doubly}

In this section, we prove our main results in the regime where the droplet is doubly connected. As one might expect from the simpler description of the droplet in this case, the overall proof and computations are significantly simpler than their counterparts in the simply connected domain, which will be addressed in the next section.   

\subsection{Description of the droplet}\label{Subsection_doubly connected droplet}
It is convenient to introduce the notations
\begin{align}
    E &:=\Bigl\{(x,y)\in \mathbb{R}^2 : \Big(\frac{x}{1+\tau}\Big)^2 + \Big(\frac{y}{1-\tau}\Big)^2 \leq 1+c
    \Bigl\},  \label{def of E ellipse} 
    \\
    D & :=\Bigl\{(x,y)\in \mathbb{R}^2:(x-p)^2+y^2 < c(1-\tau^2)\Bigl\}. \label{def of D disc} 
\end{align} 
Notice that the droplet in \eqref{droplet_doubly connected} is the same as \( E \cap D^c \).
Recall that for a given domain $\Omega$, the Schwarz function $\mathsf{S}_\Omega$, Cauchy transform $C_\Omega$, and quadrature function $r_\Omega$ are defined by \eqref{def of Schwarz ftn}, \eqref{eq_Cauchy transform}, and \eqref{def of quad function}, respectively.

Let us formulate the following lemma.
\begin{lem} \label{Lem_QD ftn for E D}
We have  
\begin{equation} \label{Sch function of E D}
\mathsf{S}_{E^c}(\zeta)  = \frac{1+\tau^2} {2\tau}\zeta - \frac{1-\tau^2}{2\tau} \sqrt{\zeta^2-4\tau(1+c)}, \qquad  \mathsf{S}_{D}(\zeta)  = p+ \frac{c(1-\tau^2)}{\zeta-p} 
\end{equation}
and 
\begin{equation} \label{eq_Cauchy E}
 C_{E}(\zeta) = \frac{1-\tau^2}{2\tau}\Big(\zeta-\sqrt{\zeta^2-4\tau(1+c)}\Big), \qquad C_{D^c}(\zeta) =p .
\end{equation}
 Consequently, we have 
 \begin{equation} \label{QD ftn E D}
    r_{E^c}(\zeta)  =  \tau \zeta, \qquad 
    r_D(\zeta)  =   \frac{c(1-\tau^2)}{\zeta-p}. 
 \end{equation}
\end{lem}
\begin{proof}
The formulas for the Schwarz functions \eqref{Sch function of E D} follow directly from straightforward computations using \eqref{def of E ellipse}, \eqref{def of D disc}, and \eqref{def of Schwarz ftn}.  
The evaluation of $C_{E}$ is given by \cite[Lemma 2.4]{By24}. For the evaluation of $C_{D^c}$, notice that for $\zeta \in D$, by using Green's formula and residue calculus, we have 
\begin{align*}\label{eq_Cauchy D^c}
    C_{D^c}(\zeta) &= \int_{D^c} \frac{1}{\zeta-\eta}\ud A(\eta) = -\frac{1}{2\pi i}\int_{\partial D}\frac{\bar{\eta}}{\zeta-\eta}\ud \eta = -\frac{1}{2\pi i}\int_{\partial D}\frac{ \mathsf{S}_{D^c}(\eta)}{\zeta-\eta}\ud \eta = p.
\end{align*}
Finally, \eqref{QD ftn E D} follows from \eqref{Sch function of E D}, \eqref{eq_Cauchy E} and \eqref{Schwarz-QD}. 
\end{proof}

By Remark~\ref{rem_uniqueness of QD}, \( D \) is the unique bounded quadrature domain with quadrature function \( r_D \) of order 1.  
Similarly, \( E^c \) is the unique unbounded quadrature domain with quadrature function \( r_{E^c} \), whose complement has area \( \pi(1+c)(1-\tau^2) \), since \( r_{E^c} \) is a polynomial of degree less than 2.

\medskip 

Using the previous lemma, we establish the following result. Recall that the potential \( Q \) is given by \eqref{eq:potential}.

\begin{prop}\label{prop_doubly connected droplet} The droplet \( S_Q \) is doubly connected if and only if \( D \subset E \). In this case, the droplet is given  by $E \cap D^c$. 
\end{prop}

\begin{proof} It has already been established in \cite[Proposition 2.1]{By24} that if the parameters \( (p, c, \tau) \) are chosen such that \( D \subset E \), then \( S = E \cap D^c \). This follows from explicit computations verifying the variational conditions \eqref{eq:variational}.  
Therefore, it suffices to prove that if the droplet is doubly connected, then $D \subset E$. 
 
Suppose the droplet \( S \) is doubly connected for given parameters \( (p, c, \tau) \). Then \( S^c \) consists of an unbounded component \( \Omega_1 \) and a bounded component \( \Omega_2 \), both of which are simply connected. By definition, it is obvious that $\Omega_1 \cap \Omega_2 = \emptyset$.  
Since \( Q \) is an algebraic potential, by Theorem~\ref{thm_S^c=QD}, \( \Omega_1 \) and \( \Omega_2 \) are both quadrature domains. Let \( r_1 \) and \( r_2 \) denote the quadrature functions of \( \Omega_1 \) and \( \Omega_2 \), respectively. Then, by \eqref{LM sum of quad functions} and \eqref{eq:quad h}, we have
\begin{equation*}
    r_1(\zeta) + r_2(\zeta) =  \tau \zeta + \frac{c(1-\tau^2)}{\zeta -p}. 
\end{equation*} 
Note that \( r_1 \) and \( r_2 \) are rational functions, with all their poles contained in their respective quadrature domains. In particular, they cannot share the same pole.  
Additionally, since \( \Omega_2 \) is a bounded quadrature domain, \( r_2(\zeta) \to 0 \) as \( \zeta \to \infty \).  
These conditions imply that
\begin{equation*}
    r_1(\zeta) = \tau \zeta, \qquad
    r_2(\zeta) = \frac{c(1-\tau^2)}{\zeta-p}.
\end{equation*}
Then, by Lemma~\ref{Lem_QD ftn for E D} and the uniqueness property discussed above, it follows that \( \Omega_2 = D \).  
Furthermore, since the area of \( S \) is \( \pi(1-\tau^2) \) (cf.~\eqref{eq:eqmeasure}), we have
\begin{align*}
    \text{area }\Omega_1^c = \text{area }S + \text{area }\Omega_2 =   \pi(1+c)(1-\tau^2).
\end{align*}
Then, once again, by Lemma~\ref{Lem_QD ftn for E D} and the uniqueness property, it follows that \( \Omega_1 = E^c \).  
Therefore by \( \Omega_1 \cap \Omega_2 = \emptyset \),  we conclude that \( D \subset E \), which completes the proof.
\end{proof} 

\subsection{Electrostatic energies}
\label{subsection doubly connected energies}

In this subsection, we compute the logarithmic energy \eqref{def of log energy}. 
Recall that the Robin's constant is given by \eqref{eq:variational}. 
Since $\mu_W$ is a probability measure, we have 
\begin{equation} \label{energy in terms of Robin}
  I_W(\mu_W) =    C_W + \frac{1}{2}\int_\mathbb{C} W(\zeta) \ud \mu_W(\zeta).
\end{equation}

\begin{lem}\label{lem_energy doubly connected}
Suppose that the droplet $S$ is doubly connected for given parameters $(p,c,\tau)$. Then the Robin's constant, denoted $C_{\rm d}(p,c,\tau)$, is evaluated as
\begin{align}\label{Robin_doubly connected}
    C_{\rm d}(p,c,\tau) &= \frac{1+c}{2}-\frac{1+c}{2}\log(1+c).
\end{align}
Furthermore, we have 
\begin{align} \label{int of Q doubly}
    \int_\mathbb{C} Q(\zeta)\ud \mu_Q(\zeta) &= \frac{1}{2} +2c+ c^2\log c(1-\tau^2) - c(1+c)\log(1+c) -\frac{2c\,p^2}{1+\tau}.
\end{align}
In particular, $\mathcal{I}_{ \rm d}(p,c,\tau)$ is given by~\eqref{weighted energy doubly connected}.
\end{lem}
 
We note that Lemma~\ref{lem_energy doubly connected} generalises previous results for \( \tau = 0 \) given in \cite[Lemma 4.8, post-critical case]{BSY24}.  
Additionally, observe that the Robin's constant does not depend on \( p \) and \( \tau \) as long as the droplet is doubly connected.

\begin{proof}[Proof of Lemma~\ref{lem_energy doubly connected}]
  
By Proposition~\ref{prop_doubly connected droplet} and \eqref{eq:variational}, the Robin's constant can be written as 
\begin{align}
    C_{\rm d}(p,c,\tau) &= \frac{1}{(1-\tau^2)}\int_{E} \log\frac{1}{|\zeta-\eta|} \ud A(\eta) - \frac{1}{(1-\tau^2)}\int_{D} \log\frac{1}{|\zeta-\eta|} \ud A(\eta) + \frac{1}{2}Q(\zeta), 
\end{align}
for $\zeta \in S= E \cap D^c$. 
By \cite[Lemma 2.4]{By24}, we have 
\begin{equation} \label{eq:log cal}
  \int_E \log \frac{1}{|\zeta-\eta|} \ud A(\eta)  = -\frac{1}{2}(|\zeta|^2-\tau \re \zeta^2) + \int_E \log \frac{1}{|\eta|}\ud A(\eta), \qquad \textup{for }\zeta \in E.
\end{equation}
On the other hand, by \cite[Lemma 2.6]{By24}, we have 
\begin{equation}  \label{eq:log cal 2} 
 \int_D \log \frac{1}{|\zeta-\eta|} \ud A(\eta)  =    -c(1-\tau^2) \log |\zeta-p|,\qquad  \textup{for } \zeta \notin D.
\end{equation}
Combining all of the above with \eqref{eq:potential}, it follows that 
\begin{align}
    C_{\rm d}(p,c,\tau) = \frac{1}{(1-\tau^2)} \int_E\log \frac{1}{|\eta|}\ud A(\eta).
\end{align}
With the elliptic coordinate $(x,y) =  (r(1+\tau)\cos \theta, r(1-\tau)\sin \theta)$, this can be rewritten as 
\begin{align*}
    C_{\rm d}(p,c,\tau) 
    &= -\frac{1}{2\pi}\int_0^{2\pi}\int_0^{\sqrt{1+c}}\log \Big(r^2\big((1+\tau)^2\cos^2 \theta + (1-\tau)^2\sin^2\theta\big)\Big)r \ud r \ud \theta\\
    &= -2\int_0^{\sqrt{1+c}} r\log r \ud r - \frac{1+c}{4\pi} \int_0^{2\pi} \log\big((1+\tau)^2\cos^2 \theta + (1-\tau)^2\sin^2\theta\big) \ud \theta. 
\end{align*} 
Then the integral evaluation (see, e.g. \cite[Eq.(4.226.6)]{GR14})
\begin{align*}
    \int_0^{\frac{\pi}{2}}\log(a^2 \cos^2\theta + b^2 \sin^2\theta) \ud \theta = \pi \log\Big(\frac{a+b}{2}\Big), \qquad (a,b>0)
\end{align*}
gives rise to the desired formula \eqref{Robin_doubly connected}.

Next, we prove \eqref{int of Q doubly}. Note that by \eqref{eq:potential} and \eqref{eq for Frostman}, we have  
\begin{align*}
    \int_\mathbb{C}Q(\zeta) \ud \mu_Q(\zeta)&= \frac{1}{(1-\tau^2)^2}\int_{S} |\zeta|^2- \tau \re \zeta^2 \ud A(\zeta) + \frac{2c}{1-\tau^2}\int_S \log \frac{1}{|\zeta-p|}\ud A(\zeta).
\end{align*}
By Green's formula and the change of variables $\zeta=\sqrt{1+c}\,(z+\tau/z)$, we have
\begin{align}
\begin{split}
    &\quad \int_{E} |\zeta|^2-\tau \re \zeta^2 \ud A(\zeta) = \frac{1}{2\pi i} \int_{\partial E} \Big( \frac12 \bar{\zeta}-\tau \zeta \Big)|\zeta|^2 \ud \zeta
    \\
    &= \frac{(1+c)^2}{2\pi i}\int_{\partial \mathbb{D}} \bigg( \frac{1}{2}\Big(z+\frac{\tau}{z}\Big)\Big(\frac{1}{z}+\tau z\Big)^2 -\tau \Big(z+\frac{\tau}{z}\Big)^2 \Big(\frac{1}{z}+\tau z\Big) \bigg) \Big(1-\frac{\tau}{z^2}\Big)\ud z  = \frac{(1+c)^2}{2}(1-\tau^2)^2.
\end{split}
\end{align}
Here, the last equality follows directly from straightforward residue calculus.  

Similarly, with the change of variables $\zeta = p+\sqrt{c(1-\tau^2)}z$, we have
\begin{align}
\begin{split}
    \int_{D} |\zeta|^2-\tau \re \zeta^2 \ud A(\zeta) = \frac{1}{2\pi i} \int_{\partial D} \Big( \frac12 \bar{\zeta}-\tau \zeta \Big)|\zeta|^2 \ud \zeta  =\Big(\frac{c^2}{2} + \frac{cp^2}{1+\tau}\Big)(1-\tau^2)^2.
\end{split}
\end{align}
On the other hand, by applying \eqref{eq:log cal}, we have 
\begin{align}
\begin{split}
    \frac{1}{1-\tau^2}\int_{S} \log \frac{1}{|\zeta-p|} \ud A(\zeta)  
    &= -\frac{p^2}{2(1+\tau)}+\frac{1}{2}-\frac{1+c}{2}\log(1+c)+ \frac{c}{2}\log c(1-\tau^2).
\end{split}
\end{align}
Combining all of the above, we obtain \eqref{int of Q doubly}. 

The last assertion immediately follows from \eqref{Robin_doubly connected} and \eqref{int of Q doubly}, together with \eqref{energy in terms of Robin}.  
\end{proof}

\section{Simply connected droplet}  \label{Section_simply}

In this section, we establish our main results in the regime where the droplet is simply connected.  

An effective approach to determining the shape of the droplet is the so-called \textit{conformal mapping method}. This method relies on an \textit{a priori} assumption about the topology of the droplet--simple connectivity in our case--and then seeks to characterise a conformal map \( f \) that maps \( \bar{\mathbb{D}}^c \) conformally onto \( S^c \).
The general procedure is as follows: First, we make use of the explicit form of the Schwarz function to extend \( f \) analytically to the entire complex plane. Then, by identifying the locations and coefficients of its poles, we establish that \( f \) is a rational function via Liouville’s theorem. Indeed, for Hele-Shaw-type potentials, where the density of the equilibrium measure is flat, it is known in general that \( f \) is a rational univalent map on \( \bar{\mathbb{D}}^c \), cf. Remark~\ref{rem_boundary of simply connected QD}. 
For implementations of this strategy in various models, see \cite{ABK21,By24,BFL25}. This method is powerful for determining the explicit shape of the droplet; however, once the conformal map is specified, one must still verify the variational conditions \eqref{eq:variational} to justify the a priori assumption about the topology of the droplet.

Previous work in this direction \cite{ABK21,By24,BFL25} has shown that, although the conformal mapping method requires separate verification of the variational conditions \eqref{eq:variational}, it nevertheless yields the correct answer for the droplet. However, this is not entirely the case in our present setting.
More precisely, the conformal mapping method provides an ansatz for the rational map but does not fully characterise the droplet, as the range of the parameter \( \kappa \) obtained through this method is larger than the true solution, see Figure~\ref{Fig_Range of kappa}. Consequently, determining the correct range, which is indeed the challenging part, becomes essential when verifying the variational conditions. 

Let us be more precise within our current setup. 
Recall that the rational map \( f \) is given by~\eqref{eq_simply conformal map}, with parameters \( (R,a,\kappa) \) satisfying the algebraic conditions~\eqref{eq:sc2}, \eqref{eq:sc3}, and \eqref{eq:sc4}. Here, it is crucial that \( a \in (0,1) \) and \( \kappa \in [0, \kappa_{\rm cri}) \), where \( \kappa_{\rm cri} \) is the unique zero of \( H(a, \cdot) \) in~\eqref{def of H(a,kappa)}. 
Beyond \( \kappa_{\rm cri} \), there exist additional critical values of \( \kappa \) that determine the geometric properties of \( f(\partial \mathbb{D}) \):
\begin{equation}\label{eq_kappa max}
      \kappa_{\rm min} = -(1-\tau)(1-a)^2, \qquad   \kappa_{\rm max} = \frac{(1-\tau a^2)^2}{(1+\tau a^2)^2}(1+\tau)(1-a^2).
\end{equation} 
Later, in the proof of Proposition~\ref{prop_simply inequality}, we will verify that
$\kappa_{ \rm cri } \le \kappa_{ \rm max }$. 
We summarise the key geometric properties of \( f(\partial \mathbb{D}) \) depending on the values of $\kappa$, along with our overall argument to establish the main results.

\begin{itemize}
    \item (Univalence condition) In general, the rational map \( f \) of the form~\eqref{eq_simply conformal map} is not univalent on \( \bar{\mathbb{D}}^c \). However, Proposition~\ref{prop_univalence} shows that \( f \) is univalent on \( \bar{\mathbb{D}}^c \) if and only if \( \kappa \in [\kappa_{\rm min}, \kappa_{\rm max}] \). Furthermore, this condition is equivalent to requiring that \( S_Q^c \) is a quadrature domain. While this step is not essential to completing the proof of our main results, it provides a necessary condition for the range of \( \kappa \).  
    \smallskip 
    \item (A priori droplet condition) Suppose that the droplet \( S_Q \) is simply connected. Then, using the conformal mapping method introduced above, Proposition~\ref{prop:a priori} shows that \( S_Q \) is enclosed by the rational map \( f \) of the form~\eqref{eq_simply conformal map}, where \( \kappa \in [0, \kappa_{\rm max}] \).  
Here, compared to the univalence condition, which requires \( \kappa \in [\kappa_{\rm min}, \kappa_{\rm max}] \), the range \( \kappa \in [\kappa_{\rm min}, 0) \) is excluded. From a computational perspective, this follows from the fact that \( c \geq 0 \), which is necessary to ensure the finiteness of the partition functions in \eqref{Gibbs complex} and \eqref{Gibbs symplectic}.  
    \smallskip 
    \item (Variational conditions) As explained above, the previous step is not necessary to fully characterise the droplet, as \( \kappa_{\rm cri} \) can be strictly smaller than \( \kappa_{\rm max} \). In the final step, we must verify that a genuine droplet arises only when \( \kappa \in [0, \kappa_{\rm cri}) \).  
For this purpose, let \( K \) be a simply connected domain enclosed by the image of the unit circle under the rational map~\eqref{eq_simply conformal map}, where \( \kappa \in [0, \kappa_{\rm max}] \). Then, in Propositions~\ref{prop_simply equality} and~\ref{prop_simply inequality}, we show that the probability measure (cf. Lemma~\ref{Lem_total mass 1})
    \begin{align} \label{def of mu K ansatz}
    \mu_K := \frac{1}{1-\tau^2}\mathbbm{1}_K \ud A  
\end{align}
    satisfies the variational conditions \eqref{eq:variational} if and only if $\kappa \in [0, \kappa_{\rm cri})$.
\end{itemize}

See Figure~\ref{Fig_Range of kappa} for an illustration and summary of this discussion.  

\begin{figure}[t]
  \begin{center}
\begin{tikzpicture}[scale=3]
    \draw[thick] (-2.25,0) -- (-1,0);
    \draw[thick] (1,0) -- (2.25,0);
    \draw[thick, red] (-1,0) -- (1,0);
    
    \filldraw (-2,0) circle (1pt) node[below=5pt] {$\kappa_{ \rm min }$};
    \filldraw (-1,0) circle (1pt) node[below left=5pt] {$ 0 $};
    \filldraw (1,0) circle (1pt) node[below=5pt] { $ \kappa_{ \rm cri } $ };
    \filldraw (2,0) circle (1pt) node[below right=5pt] { $ \kappa_{ \rm max } $ };
    
    \draw (2,0) arc[start angle=0, end angle=180, x radius=2cm, y radius=0.75cm] node[midway, above=5pt] {Univalence condition; Prop.~\ref{prop_univalence}};
    \draw (1,0) arc[start angle=0, end angle=180, x radius=1cm, y radius=0.25cm] node[midway, above=5pt] {Variational conditions; Prop.~\ref{prop_simply equality}  \& ~\ref{prop_simply inequality}};
    \draw (2,0) arc[start angle=360, end angle=180, x radius=1.5cm, y radius=0.375cm] node[midway, below=5pt] {A priori droplet condition; Prop.~\ref{prop:a priori}};
\end{tikzpicture}
\end{center}
\caption{The plot illustrates the ranges of \( \kappa \) for which different geometric properties of \( f(\partial \mathbb{D}) \) arise. } \label{Fig_Range of kappa}
\end{figure}
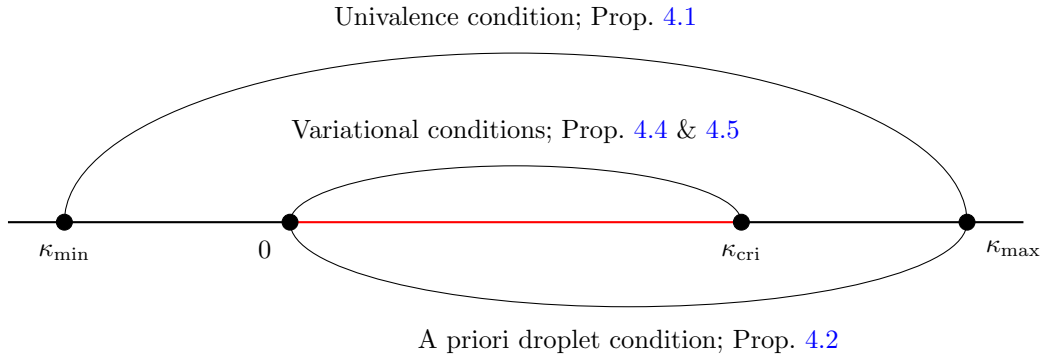

We note that \eqref{def of mu K ansatz} constitutes a slight abuse of notation, as $\mu_Q$ also denotes the equilibrium measure associated with the potential $Q$. However, since $\mu_Q$ does not appear elsewhere in this section, we believe no ambiguity should arise.

From the perspective of the phase diagram, the critical value \( \kappa_{\rm cri} \) corresponds to the intersection of Regimes II and III, marking the emergence of a new archipelago. See Figure~\ref{Fig_image of f} for images of \( f(\partial \mathbb{D}) \), where different phases can be observed depending on the values of \( \kappa \).

\begin{figure}[!tbp]
\begin{subfigure}{0.3\textwidth}
    \begin{center}
        \includegraphics[width=\textwidth]{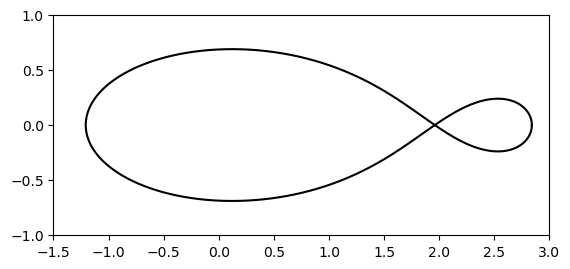}
        \subcaption{$\kappa = -0.2<\kappa_{\rm min}$}
    \end{center}
\end{subfigure}
\qquad \begin{subfigure}{0.3\textwidth}
    \begin{center}
        \includegraphics[width=\textwidth]{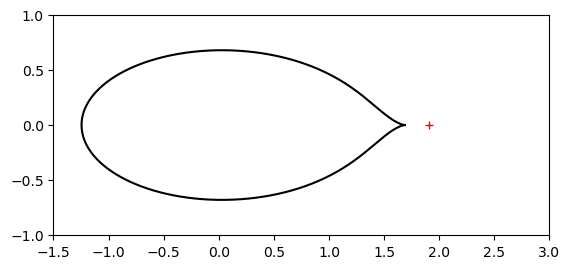}
        \subcaption{$\kappa = \kappa_{\rm min} = -0.063$}
    \end{center}
\end{subfigure}
\qquad \begin{subfigure}{0.3\textwidth}
    \begin{center}
        \includegraphics[width=\textwidth]{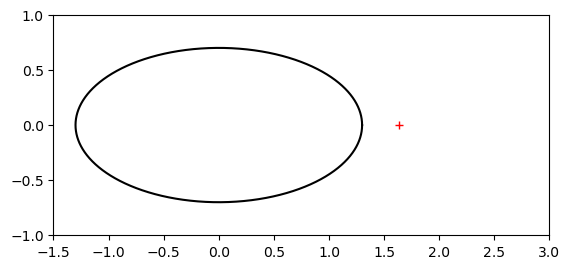}
        \subcaption{$\kappa=0$}
    \end{center}
\end{subfigure}

\begin{subfigure}{0.3\textwidth}
    \begin{center}
        \includegraphics[width=\textwidth]{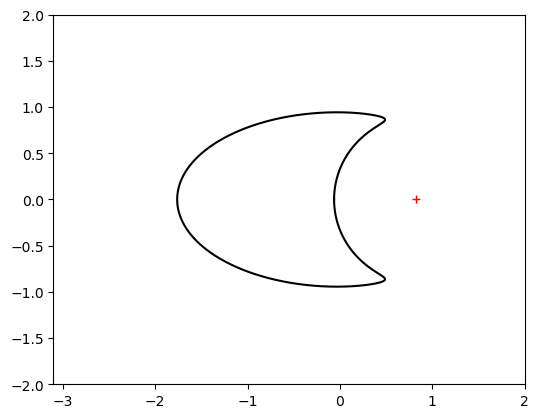}
        \subcaption{$\kappa=\kappa_{\rm cri}\approx 0.253$}
    \end{center}
\end{subfigure}
\qquad \begin{subfigure}{0.3\textwidth}
    \begin{center}
        \includegraphics[width=\textwidth]{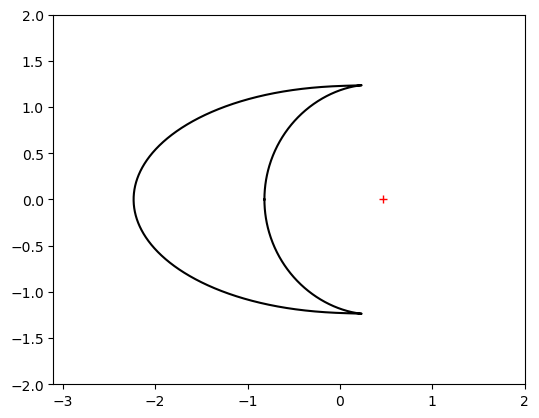}
        \subcaption{$\kappa=\kappa_{\rm max}\approx 0.367$}
    \end{center}
\end{subfigure}
\qquad \begin{subfigure}{0.3\textwidth}
    \begin{center}
        \includegraphics[width=\textwidth]{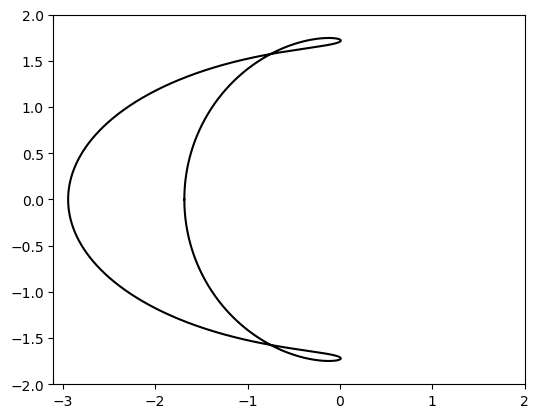}
        \subcaption{$\kappa=0.45>\kappa_{\rm max}$}
    \end{center}
\end{subfigure}
\caption{The plots show the image of \( f(\partial \mathbb{D}) \) for \( \tau = 0.3 \) and \( a = 0.7 \), with different values of \( \kappa \). The red cross in (B)–(E) indicates the point \( p \). It is clear that in (A) and (F), where \( \kappa < \kappa_{\rm min} \) or \( \kappa > \kappa_{\rm max} \), the rational map \( f \) is no longer univalent on \( \bar{\mathbb{D}}^c \).  
Plots (B) and (E) illustrate cases where univalence breaks, while plots (C) and (D) correspond to regimes where \( f(\partial \mathbb{D}) \) forms the boundary of the droplet. } \label{Fig_image of f}
\end{figure}

\subsection{Univalence criterion and conformal mapping method}
\label{subsection_simply connected droplet}

In this subsection, we discuss the a priori condition for the droplet. 

We first state the univalence criterion for the rational function $f$ of the form ~\eqref{eq_simply conformal map}. 

\begin{prop}
\label{prop_univalence}
Let \( f \) be a rational function of the form~\eqref{eq_simply conformal map}, where \( R > 0 \), \( a \in (0,1) \), and \( \tau \in [0,1) \). Then, for \( \kappa \in \mathbb{R} \), \( f \) is univalent on \( \bar{\mathbb{D}}^c \) if and only if \( \kappa \in [\kappa_{\rm min}, \kappa_{\rm max}] \).
\end{prop}

This proposition is not directly used to prove our main results. Nonetheless, it highlights an interesting feature of the rational map \( f \), and we defer the proof of Proposition~\ref{prop_univalence} to Appendix~\ref{section univalence}. We note that the proof is based on the Schur-Cohn test.

Next, we formulate the a priori condition for the droplet, which leads to the ansatz for the droplet under the assumption that it is simply connected.

\begin{prop} 
\label{prop:a priori}
Suppose the droplet \( S \) associated with the potential \( Q \) is simply connected. Then, there exists a rational map \( f \) of the form~\eqref{eq_simply conformal map} that conformally maps \( \bar{\mathbb{D}}^c \) onto \( S^c \). Here, \( (R,a,\kappa) \) satisfy the algebraic equations~\eqref{eq:sc2}, \eqref{eq:sc3}, and \eqref{eq:sc4}, where \( R > 0 \), \( a\in (0,1) \), and \( \kappa \in [0,\kappa_{\rm max}] \). 
\end{prop}

Before presenting the proof, we first comment on the positivity of the parameter \( a \). Intuitively, this may seem clear since the point \( p \), where the point charge is inserted, is non-negative, meaning that the droplet leans towards the left half-plane. However, a rigorous proof requires a finer upper bound on \( \kappa \), which is more naturally discussed after Proposition~\ref{prop_simply inequality}. Therefore, we postpone the proof of $a>0$ to the end of Section~\ref{subsection_simply connected variational}, see Remark~\ref{Rem_positivity a}.

\begin{proof}[Proof of Proposition~\ref{prop:a priori}]
By Theorem~\ref{thm_S^c=QD}, \( S^c \) is a simply connected, unbounded quadrature domain with quadrature function \( \mathsf{h} \) given by~\eqref{eq:quad h}. Since \( \mathsf{h} \) has a pole at \( p \), it follows that \( p \) does not belong to \( S \). 
Furthermore, there exists a rational conformal map that conformally maps $\bar{\mathbb{D}}^c$ onto $S^c$, cf. Remark~\ref{rem_boundary of simply connected QD}. Therefore, there is a unique rational function $f$ such that $f=R \, z +O(1)$ as $z\to \infty$ where $R>0$.  We denote by $\mathsf{S}\equiv \mathsf{S}_{S^c}$ the Schwarz function of $S^c$. Then by~\eqref{Schwarz-QD}, we have
\begin{align*}
    \mathsf{S}(\zeta) = C_S(\zeta) + \tau\zeta + \frac{c(1-\tau^2)}{\zeta-p}, \qquad \zeta \in (\interior S)^c.
\end{align*}
Recall here that the Cauchy transform \(C_S\) is defined by \eqref{eq_Cauchy transform}.  
Notice here that for $z\in \partial \mathbb{D}$, by \eqref{def of Schwarz ftn}, we have 
\begin{align}\label{eq:conformal inversion}
    f(z) = f(1/\bar{z}) = \overline{f(1/z)} = C_S(f(1/z)) + \tau f(1/z) + \frac{c(1-\tau^2)}{f(1/z)-p}.
\end{align}
Observe that the last expression is meromorphic with respect to \( z \) on \( \mathbb{D} \), with at most two poles at \( z = 0 \) and \( z = a \), where \( f(1/a) = p \). Note that such an \( a \in \mathbb{D} \setminus \{0\} \) exists uniquely since \( f \) conformally maps \( \bar{\mathbb{D}}^c \) onto \( S^c \). Moreover, due to the univalence of $f$, these poles are simple.  
Therefore, the conformal map \( f \) is of the form 
\begin{align} \label{def of f ansatz}
    f(z) = R\Big(z+ r_1 + \frac{r_2}{z} +\frac{r_3}{z-a}\Big),
\end{align}
where $r_1, r_2, r_3 \in \mathbb{R}$ and $a \in (-1,0)\cup (0,1)$ by the symmetry of the droplet along the real axis. 

Next, we determine the coefficients $r_1,r_2$, and $r_3$. Then we verify that the parameters satisfy the algebraic equations~\eqref{eq:sc2}, \eqref{eq:sc3}, and \eqref{eq:sc4}.  
To establish this, we compute the asymptotics of each term in~\eqref{eq:conformal inversion} as \( z \to 0 \) and \( z \to a \).
We first notice from \eqref{def of f ansatz} that  
\begin{align*}
\begin{split}
    f(z) = \begin{cases}
\displaystyle \frac{Rr_2}{z}+R\Big(r_1-\frac{r_3}{a^2}\Big)+R\Big(1-\frac{r_3}{a^2}\Big)z +O(z^2), &\textup{as } z\to 0,
\smallskip 
\\
\displaystyle \frac{Rr_3}{z-a}+O(1),  &\textup{as } z\to a.
\end{cases}
\end{split}
\end{align*}
Consequently, it follows from straightforward computations that 
\begin{align*}
f(1/z) & = \begin{cases}
\displaystyle \frac{R}{z}+Rr_1+R(r_2+r_3)z+O(z^2), &\textup{as } z\to 0,
\smallskip 
\\
O(1),  &\textup{as } z\to a,
\end{cases}
\\
\frac{1}{f(1/z)-p} &= \begin{cases}
\displaystyle \frac{z}{R} +O(z^2), &\textup{as } z\to 0,
\smallskip 
\\
\displaystyle -\frac{1}{R}\Big(\frac{1-\tau a^2}{a^2}-\frac{r_3}{(1-a^2)^2}\Big)^{-1}\frac{1}{z-a}+O(1),  &\textup{as } z\to a.
\end{cases}
\end{align*}
Note that by definition \eqref{eq_Cauchy transform}, we have 
\begin{align*}
    C_S(\zeta) =   \frac{1-\tau^2}{\zeta}+O(1/\zeta^2), \qquad \zeta \to \infty,
\end{align*}
where we have used $\text{area }S = \pi(1-\tau^2)$, cf. \eqref{eq:eqmeasure}. 
Then it follows that 
\begin{align*}
C_S(f(1/z)) = \begin{cases}
\displaystyle \frac{1-\tau^2}{R}z +O(z^2), &\textup{as } z\to 0,
\smallskip 
\\
\displaystyle O(1),  &\textup{as } z\to a.
\end{cases}
\end{align*} 
In consistency with~\eqref{eq_simply conformal map}, let \( r_3=-\kappa \). Then, by comparing the asymptotic behavior of both sides of~\eqref{eq:conformal inversion} as \( z \to 0 \), we obtain  
\begin{align}
    r_1 = -\frac{\kappa}{a(1-\tau)},\qquad r_2 = \tau
\end{align}
and
\begin{align}\label{eq:sc5}
    (1+c)(1-\tau^2) = R^2\Big(1-\tau^2+\frac{1+\tau a^2}{a^2}\kappa\Big).
\end{align}
Therefore, we have shown that \( f \) is of the form~\eqref{eq_simply conformal map}. 
By comparing the asymptotic behaviour of both sides of~\eqref{eq:conformal inversion} as \( z \to a \), we obtain~\eqref{eq:sc3}. Combining~\eqref{eq:sc3} and~\eqref{eq:sc5} leads to~\eqref{eq:sc2}. Finally, the condition \( f(1/a) = p \) yields~\eqref{eq:sc4}.

Now we show that \( \kappa \in [0, \kappa_{\rm max}] \). Since \( f \) is assumed to be univalent on \( \bar{\mathbb{D}}^c \), Proposition~\ref{prop_univalence} ensures that \( \kappa \in [\kappa_{\rm min}, \kappa_{\rm max}] \). 
Note that by \eqref{eq:sc3} and the requirement $c\geq 0$, we have  
\begin{equation*}
    \kappa \geq 0 \quad \text{ or }\quad \kappa \leq -\frac{1-\tau a^2}{a^2}(1-a^2)^2.
\end{equation*}
On the other hand, by \eqref{eq_kappa max} and the fact that \( \kappa > \kappa_{\min} \), the second range is not possible, which leads to \( \kappa \in [0, \kappa_{\rm max}] \). 
\end{proof}

\subsection{Variational conditions}
\label{subsection_simply connected variational}

Recall that \( K \) is a simply connected domain enclosed by the image of the unit circle under the rational map~\eqref{eq_simply conformal map}, where \( \kappa \in [0, \kappa_{\rm max}] \), and that $\mu_K$ is given by \eqref{def of mu K ansatz}.  
Our goal is to show that the ansatz $\mu_K$ is indeed the equilibrium measure associated with the potential $Q$. In other words, we aim to prove that  
\begin{equation}
K= S \qquad \textup{if } \kappa \in [0, \kappa_{\rm cri}). 
\end{equation}  
To establish this, we must exclude the range $\kappa \in [\kappa_{\rm cri}, \kappa_{\rm max}]$. As previously mentioned, this follows from verifying the variational conditions~\eqref{eq:variational}.  
In Proposition~\ref{prop_simply equality}, we prove that the equality part of~\eqref{eq:variational} holds for $\mu_Q = \mu_K$ for all $\kappa \in [0, \kappa_{\rm max}]$. On the other hand, in Proposition~\ref{prop_simply inequality}, we establish that the inequality part of~\eqref{eq:variational} holds if and only if $\kappa \in [0, \kappa_{\rm cri})$. By the uniqueness of the equilibrium measure, Propositions~\ref{prop_simply equality} and~\ref{prop_simply inequality} fully characterise the simply connected case.

Note that the case \( \tau = 0 \) was already addressed in~\cite{BBLM15} (see also Remark~\ref{Rem_droplet extremal}). Therefore, we focus on the case \( \tau \in (0,1) \), which simplifies certain aspects of the presentation.

Throughout this subsection, we assume that \( \tau \in (0,1) \), \( a \in (0,1) \), and \( \kappa \in [0, \kappa_{\rm max}] \), with \( c \) and \( p \) given by~\eqref{eq_simply connected c} and~\eqref{eq_simply connected p}.  
We first discuss that $\mu_K$ is indeed a probability measure.

\begin{lem} \label{Lem_total mass 1}
The measure \( \mu_K \) in \eqref{def of mu K ansatz} has a total mass of 1.
\end{lem}
\begin{proof}
It suffices to show that the area of $K$ is $\pi(1-\tau^2)$. By using Green's formula, we have 
\begin{align*}
    \frac{ \text{area }K }{ \pi }=  \frac{1}{2\pi i } \int_{\partial K} \bar{\zeta} \ud \zeta = \frac{1}{2 \pi i }\int_{\partial \mathbb{D}} f(1/z)f'(z) \ud z 
    =  \underset{ z=0 }{\textrm{Res}}  \Big[f(1/z)f'(z) \Big]+ 
 \underset{ z=a }{\textrm{Res}}  \Big[f(1/z)f'(z) \Big].
\end{align*}
Here, residue calculus using \eqref{eq_simply conformal map}, \eqref{eq:sc2}, and \eqref{eq:sc3} gives that
\begin{align*}
   \underset{ z=0 }{\textrm{Res}}  \Big[f(1/z)f'(z) \Big] = (1+c)(1-\tau^2), \qquad 
     \underset{ z=a }{\textrm{Res}}  \Big[f(1/z)f'(z) \Big] = -c(1-\tau^2).
\end{align*}
This completes the proof. 
\end{proof}

We define \( \mathcal{U}_K: \mathbb{C} \to (-\infty, \infty] \) by
\begin{align}\label{def of functional US}
    \mathcal{U}_K(\zeta) = \frac{1}{1-\tau^2}\int_{K} \log\frac{1}{|\zeta-\eta|^2}\ud A(\eta) + Q(\zeta).
\end{align}
This is the left hand side of \eqref{eq:variational}, up to multiplicative constant. 
Notice that \( \mathcal{U}_K \) is a continuous function that diverges to infinity as \( |\zeta| \to \infty \) and at \( \zeta = p \) if $c>0$.

\begin{prop}\label{prop_simply equality}
For $\zeta \in K$, we have $\mathcal{U}_K(\zeta)=\ell_K$ for some constant $\ell_K$. 
\end{prop}
 
By Proposition~\ref{prop_simply equality}, we can take  
\begin{equation} \label{def of lK Robin 2}
\ell_K =  \mathcal{U}_K(f(1)). 
\end{equation} 
Notice that once \( K \) is proven to be the droplet, this value coincides with twice the Robin's constant.

\begin{proof}[Proof of Proposition~\ref{prop_simply equality}]
Since $K$ is a simply connected subset of $\mathbb{C}$, it is enough to show that the derivative $\partial_\zeta \mathcal{U}_K(\zeta)$ vanishes in the interior of $K$. Applying Green's formula and change of variables $\zeta =f(z)$, we have 
\begin{align}\label{eq_deriv of functional US}
\begin{split}
    (1-\tau^2)\partial \mathcal{U}_K(\zeta) &= -\int_K \frac{1}{\zeta-\eta}\ud A(\eta) +  \bar{\zeta}- \tau \zeta - \frac{c(1-\tau^2)}{\zeta-p} 
    \\
    &= -\frac{1}{2\pi i}\int_{\partial K}\frac{\bar{\eta}}{\zeta-\eta} \ud\eta - \bar{\zeta} \cdot \mathbbm{1}_{ \{\zeta \in \interior K\} } + \bar{\zeta}- \tau \zeta - \frac{c(1-\tau^2)}{\zeta-p} 
    \\
    &= \frac{1}{2\pi i}\int_{\partial \mathbb{D}}\frac{f(1/z)f'(z)}{f(z)-\zeta} \ud z + \bar{\zeta} \cdot \mathbbm{1}_{ \{ \zeta \notin \interior K \} } - \tau \zeta - \frac{c(1-\tau^2)}{\zeta-p}.
\end{split}
\end{align}
Notice that by \eqref{eq_simply conformal map}, the integrand in the line contains its poles in $\mathbb{D}$ at $0, a$, and $f^{-1}(\zeta)\cap \mathbb{D}$.
Since the equation $f(z)=\zeta$ is equivalent to
\begin{equation}\label{eq:f cubic equation}
    z^3 - \Big(a+\frac{\kappa}{a(1-\tau)} + \frac{\zeta}{R}\Big)z^2 + \Big(\tau + \frac{\tau\kappa}{1-\tau} + \frac{a\zeta}{R}\Big)z -\tau a=0, 
\end{equation}
for given $\zeta \in \mathbb{C}$, there exist $z_\zeta^{(j)}$ $(j=1,2,3)$ such that $f(z_\zeta^{(j)}) = \zeta$. 

If \( \zeta \in \interior K \), since \( f \) is a conformal mapping from \( \bar{\mathbb{D}}^c \) to \( K^c \), all \( z_\zeta^{(j)} \) are contained in \( \mathbb{D} \).
Then for $\zeta \in \interior K$, we have 
\begin{equation} \label{integral residue 0 a 123}
 \frac{1}{2\pi i}\int_{\partial \mathbb{D}}\frac{f(1/z)f'(z)}{f(z)-\zeta} \ud z =  \underset{ z=0 }{\textrm{Res}}  \Big[  \frac{f(1/z)f'(z)}{f(z)-\zeta} \Big]+ \underset{ z=a }{\textrm{Res}}  \Big[  \frac{f(1/z)f'(z)}{f(z)-\zeta} \Big] + \sum_{j=1}^3  \underset{ z=z_\zeta^{(j)} }{\textrm{Res}}  \Big[ \frac{f(1/z)f'(z)}{f(z)-\zeta} \Big]. 
\end{equation}
By straightforward computations, we have
\begin{align*}
\frac{f(1/z)f'(z)}{f(z)-\zeta} = \begin{cases}
\displaystyle -\frac{R}{z^2} -\frac{\zeta}{\tau}\frac{1}{z}+ O(1), &\textup{as } z\to 0,
\smallskip 
\\
\displaystyle -\frac{p}{z-a}+ O(1),  &\textup{as } z\to a, 
\end{cases}
\end{align*} 
which leads to 
\begin{align} \label{residues at 0 a}
   \underset{ z=0 }{\textrm{Res}} \Big[ \frac{f(1/z)f'(z)}{f(z)-\zeta} \Big] = -\frac{\zeta}{\tau}, \qquad  \underset{ z=a }{\textrm{Res}} \Big[  \frac{f(1/z)f'(z)}{f(z)-\zeta} \Big] = -p.
\end{align}
Using \eqref{eq:f cubic equation}, observe that 
\begin{align}\label{eq:rootsrelation}
\begin{split}
    z_\zeta^{(1)} + z_\zeta^{(2)} + z_\zeta^{(3)} &= a+\frac{\kappa}{a(1-\tau)} + \frac{\zeta}{R}, 
    \\
    z_\zeta^{(1)}z_\zeta^{(2)}+z_\zeta^{(2)}z_\zeta^{(3)}+ z_\zeta^{(3)}z_\zeta^{(1)} &= \tau + \frac{\tau\kappa}{1-\tau}+\frac{a\zeta}{R},
\end{split}
\begin{split}
    z_\zeta^{(1)}z_\zeta^{(2)}z_\zeta^{(3)} &= \tau a,\\
    \frac{1}{z_\zeta^{(1)}}+ \frac{1}{z_\zeta^{(2)}}+\frac{1}{z_\zeta^{(3)}}&= \frac{1}{a}+\frac{\kappa}{a(1-\tau)}+ \frac{\zeta}{\tau R}.
\end{split}
\end{align}
Notice also that 
\begin{equation}
 \underset{ z=z_\zeta^{(j)} }{\textrm{Res}}  \Big[ \frac{f(1/z)f'(z)}{f(z)-\zeta} \Big] = f(1/z_\zeta^{(j)}). 
\end{equation}
By using \eqref{eq_simply conformal map} and \eqref{eq:rootsrelation}, we have
\begin{align*}
    &\quad \sum_{j=1}^3 f(1/z_\zeta^{(j)}) =R\sum_{j=1}^3\bigg( \frac{1}{z_\zeta^{(j)}} +\tau z_\zeta^{(j)}-\frac{\kappa z_\zeta^{(j)}}{1-az_\zeta^{(j)}}-\frac{\kappa}{a(1-\tau)}\bigg) \\
    &=R\bigg[ -\frac{3\tau \kappa}{a(1-\tau)}+ \sum_{j=1}^3 \bigg(\frac{1}{z_\zeta^{(j)}}+\tau z_\zeta^{(j)} - \frac{\kappa }{a(1-az_\zeta^{(j)})}\bigg) \bigg] = \Big(\tau+\frac{1}{\tau}\Big)\zeta + p + R \bigg(\frac{(2-a^2)\kappa}{a(1-a^2)}-\sum_{j=1}^3\frac{\kappa }{a(1-az_\zeta^{(j)})}\bigg).
\end{align*}
Furthermore, using~\eqref{eq:sc3} we have 
\begin{align*}
&\quad \frac{(2-a^2)\kappa}{a(1-a^2)}-\sum_{j=1}^3\frac{\kappa }{a(1-az_\zeta^{(j)})}  = \frac{ \kappa}{a\prod_{j=1}^3(1-az_\zeta^{(j)})}
\\
&\quad \times \bigg(\frac{-1+2a^2}{1-a^2}-\frac{a^3}{1-a^2}(z_\zeta^{(1)}+z_\zeta^{(2)}+z_\zeta^{(3)})+\frac{a^2}{1-a^2}(z_\zeta^{(1)}z_\zeta^{(2)}+z_\zeta^{(2)}z_\zeta^{(3)}+z_\zeta^{(3)}z_\zeta^{(1)})-\frac{a^3(2-a^2)}{1-a^2}z_\zeta^{(1)}z_\zeta^{(2)}z_\zeta^{(3)}\bigg)
\\
&= - \frac{Ra^2}{(1-a^2)(f(1/a)-\zeta)}\frac{(1-\tau a^2)(1-a^2)^2\kappa +a^2\kappa^2}{a^4(1-a^2)} =\frac{1}{R}\frac{c(1-\tau^2)}{\zeta-p}. 
\end{align*}
Combining all of the above, we obtain 
\begin{equation} \label{eq:rescomp}
 \sum_{j=1}^3 f(1/z_\zeta^{(j)}) =  \Big(\tau+\frac{1}{\tau}\Big)\zeta + p + \frac{c(1-\tau^2)}{\zeta-p}.
\end{equation} 
Therefore by using \eqref{integral residue 0 a 123}, \eqref{residues at 0 a} and \eqref{eq:rescomp}, we obtain $  \partial \mathcal{U}_K(\zeta)=0 $ for  $\zeta \in \interior K$. This completes the proof. 
\end{proof}

Recall that $\kappa_{ \rm cri } \in [0,\kappa_{\rm max}] $ is a unique zero  of $H(a,\cdot)$ in~\eqref{def of H(a,kappa)}.

\begin{prop}\label{prop_simply inequality}
Let \( \zeta \in K^c \). Then the inequality \(\mathcal{U}_K(\zeta) > \ell_K \) holds if and only if \( \kappa \in [0, \kappa_{\rm cri}) \). 
\end{prop}

We prove Proposition~\ref{prop_simply inequality} by breaking the argument into several steps. By definition \eqref{def of functional US}, it is evident that \( \mathcal{U}_K \) attains its global minimum on \( \mathbb{C} \setminus \{p\} \). Therefore, it suffices to show that any local minima outside \( K \) have values greater than \( \ell_K \). 

We first characterise the critical points of $\mathcal{U}_K$ located outside $K$. 
For $\zeta \in K^c$, let $z_\zeta^{(1)}=F(\zeta)$ where $F:K^c\to \bar{\mathbb{D}}^c$ is the conformal inverse of $f$. Then $z_\zeta^{(2)}$ and $z_\zeta^{(3)}$ are contained in $\mathbb{D}$. This in turn implies that for $\zeta \in K^c\setminus\{p\}$, by \eqref{eq_deriv of functional US} and Proposition~\ref{prop_simply equality}, we have  
\begin{align*}
 (1-\tau^2) \partial \mathcal{U}_K(\zeta) &= \bar{\zeta}-\tau\zeta -\frac{c(1-\tau^2)}{\zeta-p}+ \underset{ z=0 }{\textrm{Res}}  \Big[  \frac{f(1/z)f'(z)}{f(z)-\zeta} \Big]+ \underset{ z=a }{\textrm{Res}}  \Big[  \frac{f(1/z)f'(z)}{f(z)-\zeta} \Big] + \sum_{j=2}^3  \underset{ z=z_\zeta^{(j)} }{\textrm{Res}}  \Big[ \frac{f(1/z)f'(z)}{f(z)-\zeta} \Big]
 \\
 &= \bar{\zeta} - f(1/z_\zeta^{(1)}). 
\end{align*} 
Since the Schwarz function $\mathsf{S}(\zeta)\equiv \mathsf{S}_{K^c}(\zeta)$ of $K^c$ is given by 
\begin{equation} \label{Schwarz in terms of f F}
\mathsf{S}(\zeta) = f(1/F(\zeta)), 
\end{equation}
we have shown that 
\begin{align} \label{eq_functional US derivative}
    (1-\tau^2)\partial\mathcal{U}_K(\zeta) =
    \begin{cases}
        0 & \zeta \in \interior K,
        \smallskip 
        \\
        \bar{\zeta}- \mathsf{S}(\zeta) & \zeta \in K^c\setminus \{p\}.
    \end{cases}
\end{align}
Thus, the characterisation of critical points reduces to finding \( z \) such that
\begin{equation} \label{eqn for cri pts}
f(\bar{z}) = f(1/z), \qquad |z|>1. 
\end{equation}  
Define
\begin{align} \label{def of kappa1}
    \kappa_1 := \frac{(1-\tau)^2}{(1+\tau)}(1-a^2). 
\end{align}
Note that $\kappa_1 \le \kappa_{ \rm \max }$, where $\kappa_{ \rm max }$ is given by \eqref{eq_kappa max}. Here, notice that the equality holds when $\tau=0$. 

\begin{lem}\label{lem:three crit}
The critical points of \( \mathcal{U}_K \) that lie outside \( K \) are given as follows.
\begin{itemize}
    \item[\textup{(i)}] \textup{(Single critical point regime)} If \( \kappa \in (0, \kappa_{1}] \cup \{ \kappa_{\rm max } \}  \), there is exactly one real critical point located in the interval \( (p, \infty) \).
    \smallskip
    \item[\textup{(ii)}] \textup{(Three critical points regime)} If \( \kappa \in (\kappa_1, \kappa_{\rm max}) \), in addition to the real critical point in \( (p, \infty) \), there exist two distinct non-real conjugate critical points.
\end{itemize} 
\end{lem}

\begin{proof}
Let us write 
\begin{align}\label{def of rational g}
    g(z) = z+ \frac{\tau}{z}-\frac{\kappa}{z-a}. 
\end{align} 
Since $f$ is given by \eqref{eq_simply conformal map}, we have 
\begin{equation} \label{rel of f and g}
f(z) =R\Big(g(z)-\frac{\kappa}{a(1-\tau)}\Big). 
\end{equation}
Thus the identity \eqref{eqn for cri pts} is equivalent to $g(\bar{z})=g(1/z)$. 

Notice that if \( \kappa = 0 \), the equation \( g(\bar{z}) = g(1/z) \) simplifies to \( \bar{z} = \tau z \), which has no roots with an absolute value greater than 1.

From now on we consider the case $\kappa \in (0,\kappa_{\rm max}]$. Suppose that $z \in \mathbb{R}$. Then $g(z)=g(1/z)$ reads as
\begin{equation}\label{eq_real critical point equation}
    z+ \frac{1}{z} = a+ \frac{1}{a}+\frac{\kappa}{a(1-\tau)}.
\end{equation}
By solving this equation, let 
\begin{equation} \label{def of zstar}
z_* := \frac1{2a} \bigg( a^2+ 1+\frac{\kappa}{1-\tau} + \sqrt{ \Big( a^2+ 1+\frac{\kappa}{1-\tau} \Big)^2-4a^2 }\,\bigg). 
\end{equation}
Then $z_* >1/a$, $g(z)=g(1/z)$, and $f(z_*) \in (p,\infty)$ is a real critical point lying outside $K$.

Assume that a non-real $|z|>1$ satisfies $g(\bar{z})=g(1/z)$. In terms of the polar coordinates $z= re^{i\theta}$, we have 
\begin{align*}
 (1-\tau a^2) \cos 2\theta = \frac{1+r^2}{r} a(1-\tau) \cos \theta -\kappa +\tau -a^2, \qquad   (1+\tau a^2) \sin 2\theta = \frac{1+r^2}{r}a(1+\tau)\sin \theta.
\end{align*} 
Due to the assumptions $r>1$ and  $\sin \theta \neq 0$, these equations are equivalent to
\begin{align*}
    \cos^2\theta = \frac{1+\tau}{4\tau}\Big(1+\tau -\frac{\kappa}{1-a^2}\Big) \in [0,1),\qquad
    \frac{1+r^2}{2r} =\frac{1+\tau a^2}{a(1+\tau)}\cos \theta >1,
\end{align*}
which admit solutions $(r,\theta)$ if and only if $\kappa \in (\kappa_1, \kappa_{\rm max})$. If it is the case, there are precisely two conjugate non-real critical points of $\mathcal{U}_K$ outside $K$.
\end{proof}

We denote the real critical point of \( \mathcal{U}_K \) in \( (p, \infty) \) by \( \zeta_* \equiv \zeta_*(a, \kappa, \tau) \) and its conformal preimage by 
\begin{equation} \label{def of zeta star}
z_* = F(\zeta_*). 
\end{equation}
Note that \( z_* \in (1/a, \infty) \) satisfies~\eqref{eq_real critical point equation}.
As \( \kappa \) increases from \( \kappa_1 \) to \( \kappa_{\rm max} \), the non-real critical points of \( \mathcal{U}_K \) outside \( K \) move away from \( \zeta_*(a, \kappa_1, \tau) \) and approach \( \partial K \). 

Before examining whether the critical points identified above are local minima, we first establish that all points in \( K \) are local minima of \( \mathcal{U}_K \) when \( \kappa < \kappa_{\rm max} \).
Namely, we claim that there exists an open neighbourhood \( V \) of \( K \) such that for any \( \zeta \in V \setminus K \), we have \( \mathcal{U}_K(\zeta) > \ell_K \).
We note that the following argument is essentially identical to that in~\cite[Lemma 2.2]{BBLM15}.

Since the Schwarz function \( \mathsf{S} \) can be analytically extended to an open neighbourhood of \( (\interior K)^c \) when \( \kappa < \kappa_{\rm max} \), we define
\begin{align}\label{eq:UtildeS}
    \widetilde{\mathcal{U}}_K(\zeta) := \frac{1}{1-\tau^2} \Big( |\zeta|^2 - |f(1)|^2 - 2\re \int_{f(1)}^\zeta \mathsf{S}(\eta) \, \ud \eta \Big) + \ell_K
\end{align}
on an open neighbourhood of \( (\interior K)^c \), where it coincides with \( \mathcal{U}_K \).
Observe that by~\eqref{eq_functional US derivative}, 
$$
(1-\tau^2) \partial \widetilde{\mathcal{U}}_K(\zeta) = \bar{\zeta} - \mathsf{S}(\zeta) , \qquad (1-\tau^2) \bar{\partial} \partial \widetilde{\mathcal{U}}_K = 1.
$$
For \( \zeta \in \partial K \), let \( \mathsf{n} \) be the unit normal vector at \( \zeta \) pointing outward from \( \partial K \). Then, since \( \widetilde{\mathcal{U}}_K (\zeta) = \ell_K \) along \( \partial K \), the gradient \( \mathsf{grad} \; \widetilde{\mathcal{U}}_K(\zeta) \) is parallel to \( \mathsf{n} \). Furthermore, the determinant of the Hessian of \( \widetilde{\mathcal{U}}_K  \) vanishes since \( \widetilde{\mathcal{U}}_K \) is constant along \( \partial K \).
On the other hand, the trace of the Hessian of \( \widetilde{\mathcal{U}}_K(\zeta) \) is given by  
\[
4\bar{\partial} \partial \widetilde{\mathcal{U}}_K(\zeta) = \frac{4}{1-\tau^2} \bar{\partial} (\bar{\zeta} - \mathsf{S}(\zeta)) = \frac{4}{1-\tau^2} > 0,
\]
which implies that \( \widetilde{\mathcal{U}}_K \) attains local minima at \( \zeta \in \partial K \). Consequently, due to the compactness of \( \partial K \), there exists an open neighbourhood \( V \) of \( K \) such that \( \mathcal{U}_K = \widetilde{\mathcal{U}}_K > \ell_K \) on \( V \setminus K \).

Next, we determine the local minima of $\mathcal{U}_K$ lying outside $K$. Recall that $\zeta_*$ is defined by~\eqref{def of zeta star}.

\begin{lem}\label{lem:critical points}
The function \( \mathcal{U}_K \) has local minima outside \( K \) if and only if \( \kappa \in (\kappa_1, \kappa_{\rm max}] \). In this case, \( \zeta_* \in (p, \infty) \) is the unique local minimum.
\end{lem}

\begin{proof}
When \( \kappa = 0 \), there are no critical points in \( K^c \), so the proof is complete.

Now, suppose \( \kappa \in (0, \kappa_1] \). In this case, \( \zeta_* \) is the unique critical point of \( \mathcal{U}_K \) outside \( K \). If \( \zeta_* \) were a local minimum in \( K^c \), the mountain pass theorem would guarantee the existence of another critical point outside \( K \), leading to a contradiction.
To provide further details, we follow the standard argument used, for instance, in~\cite[Lemma 2.3]{BBLM15}. Consider continuous paths from a fixed point on \( \partial K \), say \( f(1) \in \partial K \), to the local minimum \( \zeta_* \), ensuring that the paths do not pass through \( p \).
Since there exists an open neighbourhood \( V \) where \( \mathcal{U}_K(\zeta) > \ell_K \) for all \( \zeta \in V \setminus K \), and since \( \zeta_* \) is a local minimum, the maximum value of \( \mathcal{U}_K(\zeta) \) along these paths is attained at neither the starting nor the endpoint.
Taking the minimum of all such maximum values, we obtain a point \( \zeta_1 \in K^c \setminus \{ p, \zeta_* \} \) where the min–max value is achieved. A standard variational argument then shows that \( \zeta_1 \) is also a critical point, contradicting Lemma~\ref{lem:three crit}.

Next, we consider the case $\kappa \in (\kappa_1, \kappa_{\rm max}]$ and show that $\zeta_*$ is a local minimum. As shown above, the Hessian of $\mathcal{U}_K$ has trace $4/(1-\tau^2)$ and determinant
\begin{align*}
    4\Big( (\bar{\partial}\partial \mathcal{U}_K(\zeta_*))^2 -(\partial^2 \mathcal{U}_K(\zeta_*))(\bar{\partial}^2\mathcal{U}_K(\zeta_*))\Big) = \frac{4}{(1-\tau^2)^2}(1-|\mathsf{S}'(\zeta_*)|^2).
\end{align*}
Since $\mathsf{S}(\zeta)=f(1/F(\zeta))$ for $\zeta \in K^c$, we have 
\begin{align*}
    \mathsf{S}'(\zeta_*) = -f'(1/F(\zeta_*))\frac{F'(\zeta_*)}{F^2(\zeta_*)}=-\frac{f'(1/z_*)}{z_*^2 f'(z_*)}.
\end{align*}
By using the fact that $z_*$ solves \eqref{eq_real critical point equation}, we have 
\begin{align*}
    z_*^2f'(z_*) -f'(1/z_*) &= R(z_*^2-1)\Big(1+\tau-\frac{\kappa (1-a^2)z_*^2}{(z_*-a)^2(1-az_*)^2} \Big) = R(1+\tau)(z_*^2-1)\frac{\kappa-\kappa_1}{\kappa}>0,
    \\
    z_*^2f'(z_*) +f'(1/z_*) &= R\Big((1-\tau)(1+z_*^2) + \frac{\kappa z_*^2}{(z_*-a)^2}+\frac{\kappa z_*^2}{(1-az_*)^2}\Big)>0.
\end{align*}
Therefore $|\mathsf{S}'(\zeta_*)|<1$, which implies that the Hessian is a positive definite matrix. Thus $\zeta_*$ is the unique local minimum of $\mathcal{U}_K$ outside $K$.

Finally, we show that the non-real critical points of \( \mathcal{U}_K \) outside \( K \), which arise when \( \kappa \in (\kappa_1, \kappa_{\rm max}) \), are not local minima.
Indeed, at least one of these non-real critical points cannot be a local minimum, as the mountain pass argument guarantees the existence of a saddle point.
Moreover, since the non-real critical points are conjugates and \( \mathcal{U}_K \) is symmetric with respect to the real axis, we conclude that \( \zeta_* \) is the unique local minimum when \( \kappa \in (\kappa_1, \kappa_{\rm max}] \).
\end{proof}

We are now ready to complete the proof of Proposition~\ref{prop_simply inequality}.

\begin{proof}[Proof of Proposition~\ref{prop_simply inequality}]
We have established that if \( \kappa \in [0, \kappa_1] \), the only local minima of \( \mathcal{U}_K \) are the points in \( K \), ensuring that the variational conditions are satisfied in this case.
Furthermore, for \( \kappa \in [0, \kappa_1] \), we have \( \mathcal{U}_K(\zeta_*) > \ell_K \).

We claim that \( \mathcal{U}_K(\zeta_*) \leq \ell_K \) when \( \kappa = \kappa_{\rm max} \).
By Lemma~\ref{lem:critical points}, for \( \kappa \in (\kappa_1, \kappa_{\rm max}) \), the value of \( \mathcal{U}_K \) at the non-real critical points outside \( K \) is greater than \( \mathcal{U}_K(\zeta_*) \).
By the continuity of \( \mathcal{U}_K \), its value at the non-real critical points converges to \( \ell_K \) as these points approach \( \partial K \) when \( \kappa \to \kappa_{\rm max} \). Thus
\begin{align*}
    \mathcal{U}_K(\zeta_*)\big|_{\kappa = \kappa_{\rm max}} = \lim_{\kappa \to \kappa_{\rm max}} \mathcal{U}_K(\zeta_*) \leq \ell_K\big|_{\kappa=\kappa_{\rm max}},
\end{align*}
which proves the claim.

It remains to show that there exists \( \kappa_{\rm cri} \in (\kappa_1, \kappa_{\rm max}] \) such that \( \mathcal{U}_K(\zeta_*) > \ell_K \) if and only if \( \kappa \in [0, \kappa_{\rm cri}) \). Recall that the Schwarz function $\mathsf{S}$ is given by \eqref{Schwarz in terms of f F}. 
For \( \zeta \in K^c \) and \( z = F(\zeta) \), applying~\eqref{eq:UtildeS} and integrating by parts, we obtain
\begin{align}\label{eq_functional computation 1}
\begin{split}
    \mathcal{U}_K(\zeta) -\ell_K  &= \frac{1}{1-\tau^2}\Big(|\zeta|^2 - \re \zeta\,\mathsf{S}(\zeta) - \re \int_{1/z}^{z} f(1/w)f'(w) \ud w\Big).
\end{split}
\end{align}
Notice that by \eqref{eq_simply conformal map}, we have  
\begin{align*}
    f(1/w)f'(w) = R^2\Big(\frac{1}{w}+\tau w -\frac{\kappa w}{1-aw}- \frac{\kappa}{a(1-\tau)}\Big)\Big(1-\frac{\tau}{w^2}+\frac{\kappa}{(w-a)^2}\Big).
\end{align*}
Then by evaluating the integral in \eqref{eq_functional computation 1} together with \eqref{eq_simply connected c} and~\eqref{eq_simply connected p}, we obtain 
\begin{align}\label{eq_functional computation 2}
\begin{split}
    \mathcal{U}_K(\zeta)-\ell_K &= \frac{1}{1-\tau^2}\Big(|f(z)|^2 - \re f(z)f(1/z) +\frac{ R\kappa p ( z^2-1) }{(z-a)(az-1)}\Big)
 - 2c \log \frac{|az-1|}{|z-a|} - 2(1+c)\log |z|.
\end{split}
\end{align}
We define 
\begin{align}\label{def of H(a,kappa)}
\begin{split}
    H(a,\kappa) &:= \frac{1-\tau}{a} \Big(1+\tau a^2 -\frac{1-\tau a^2}{1-\tau}\frac{\kappa}{1-a^2} \Big)\Big(z_*-\frac{1}{z_*}\Big)
    \\
    &\quad - 2\Big(\frac{1-\tau a^2}{a^2}\kappa + \frac{\kappa^2}{(1-a^2)^2}\Big)\log\frac{|az_*-1|}{|z_*-a|} - 2\Big(1-\tau^2 +\frac{1+\tau a^2}{a^2}\kappa\Big)\log |z_*|,
\end{split}
\end{align}
where $z_*$ is given by \eqref{def of zstar}.  
Notice that by \eqref{def of zeta star}, we have 
\begin{align} 
\begin{split}
    H(a,\kappa)  = \frac{1-\tau^2}{R^2}(\mathcal{U}_K(\zeta_*)-\ell_K) . 
\end{split}
\end{align}  

Since \( H(a, \kappa) \) is continuous, our claim reduces to verifying that \( H(a, \cdot) \) has a unique zero in \( (\kappa_1, \kappa_{\rm max}] \).
The existence of a zero is ensured by the fact that 
$$ \min_{\kappa \in [0,\kappa_1]}H(a, \kappa) > 0 \geq H(a, \kappa_{\rm max}) ,$$
as shown above.
Thus, it suffices to prove that \( H(a, \kappa) \) is concave with respect to \( \kappa \) in the range \( [0, \kappa_{\rm max}] \).

By differentiating $H(a,\kappa)$ with respect to $\kappa$, we have 
\begin{align*}
    \frac{\partial^3}{\partial \kappa^3}H(a,\kappa) = \frac{2\big((1-\tau)(1-a^2)(1-\tau a^2) + (1+\tau a^2)\kappa\big)}{a^2\kappa^2 \sqrt{\big((1-\tau)(1+a^2) + \kappa\big)^2 - 4(1-\tau)^2 a^2}}>0, \qquad \kappa > 0.
\end{align*}
On the other hand, after lengthy but straightforward computations, one can observe that
\begin{align*}
    \lim_{\kappa \to \infty} \frac{\partial^2}{\partial \kappa^2 }H(a,\kappa) = -\frac{2(1-\tau a^2)}{(1-\tau)a^2(1-a^2)} - \frac{4\log a}{(1-a^2)^2}\le -\frac{2\tau}{(1-\tau)a^2}<0.
\end{align*}
Here, we have used the fact that \(-2a^2\log a\le (1-a^2).\) Hence, \( H(a, \kappa) \) is concave with respect to \( \kappa \) in the range \( [0, \kappa_{\rm max}] \), implying that \( H(a, \cdot) \) has a unique zero \( \kappa_{\rm cri} \in (\kappa_1, \kappa_{\rm max}] \).
As a result, the inequality part of the variational condition~\eqref{eq:variational} holds if and only if \( \kappa \in [0, \kappa_{\rm cri}) \).
\end{proof} 

In the proof of Proposition~\ref{prop_simply inequality}, we have shown that 
\begin{equation}
\kappa_1 \le \kappa_{ \rm cri } \le \kappa_{ \rm max }. 
\end{equation}
As previously mentioned below \eqref{def of kappa1}, for $\tau=0$, we have $\kappa_1=\kappa_{ \rm max }$. This in turn implies that in the extremal case $\tau=0$, there is no additional phase transition of the droplet yielding the multi-component regime.

\begin{rem}[Positivity of the parameter $a$] \label{Rem_positivity a} 
Set $\kappa=(1-\tau)(1-a^2)$. Then $z_* = (1+\sqrt{1-a^2})/a$ and
\begin{align*}
     H(a, (1-\tau)(1-a^2))  &= 4\tau(1-\tau)\Big(\sqrt{1-a^2}-(2-a^2)\log \frac{1+\sqrt{1-a^2}}{a}\Big)\leq 0,
\end{align*}
which yields that $\kappa_{\rm cri}\leq (1-\tau)(1-a^2)$.  
Suppose that \( a \in (-1,0) \). Following the same steps as in Propositions~\ref{prop_simply equality} and~\ref{prop_simply inequality}, we arrive at the conclusion that the values of \( c \) and \( p \) given by~\eqref{eq_simply connected c} and~\eqref{eq_simply connected p} induce a simply connected droplet \( S \) if and only if \( \kappa \in [0, \kappa_{\rm cri}) \), where \( \kappa_{\rm cri}(a) = \kappa_{\rm cri}(-a) \) since \( H(a,\kappa) = H(-a,\kappa) \).  
The symmetry breaks at \( p \geq 0 \) because~\eqref{eq:sc4} and the condition \( \kappa < (1-\tau)(1-a^2) \) imply  
\begin{equation*}
    0\leq p = \frac{R}{a} \Big(1+\tau a^2 - \frac{1-\tau a^2}{1-\tau} \frac{\kappa}{1-a^2} \Big) < \frac{R}{a} \Big(1+\tau a^2 -(1-\tau a^2)\Big) = 2R\tau a \leq 0,
\end{equation*}
which is a contradiction. Thus, we conclude that if \(S  \) is simply connected, then \( a \in (0,1) \).
\end{rem}

\begin{rem} 
\label{rem_uniqueness problem}
In general, it is not clear whether the solutions \( (a, \kappa) \) of the algebraic equations~\eqref{eq:sc2}, \eqref{eq:sc3}, and~\eqref{eq:sc4} exist or are unique for given parameters \( (p, c, \tau) \).
However, if we assume that the parameters \( (p, c, \tau) \) correspond to a simply connected droplet \( S \), Propositions~\ref{prop_univalence} and ~\ref{prop:a priori} guarantee the existence of a solution with \( a \in (0,1) \) and \( \kappa \in [0, \kappa_{\rm max}] \).
Furthermore, Propositions~\ref{prop_simply equality} and~\ref{prop_simply inequality} imply that \( \kappa \in [0, \kappa_{\rm cri}) \).
Regarding uniqueness, the uniqueness of the equilibrium measure ensures that no two pairs \( (a, \kappa) \) correspond to the same set of parameters \( (p, c, \tau) \). This follows from the fact that each pair \( (a, \kappa) \) induces a distinct conformal map \( f \), as can be verified by examining the poles and residues.

In the extremal case \( \tau = 0 \), the existence and uniqueness problem can be addressed more explicitly. In~\cite[Appendix A]{BBLM15}, the authors examined the existence problem using the discriminant of~\eqref{cubic eqn for tau0} and addressed uniqueness by selecting the smallest nonzero root of~\eqref{cubic eqn for tau0}.
These considerations suggest that the algebraic equations arising from the conformal mapping method require additional conditions or further information to fully determine the droplet. 
\end{rem}

\begin{rem}
\label{rem_intersection regime II and III}
Here, we present a detailed exposition of two limits: \( c \to \infty \) with fixed \( p > 0 \), and \( p \to \infty \) with fixed \( c \geq 0 \) when \( \tau \in (0,1) \).  
Heuristically, in both cases, the parameters \( (p,c,\tau) \) will eventually fall within Regime II, as discussed in Remark~\ref{rem_extermal tau=1}.

Consider Regime II as the union of disjoint curves \( \kappa \mapsto (p(a,\kappa), c(a,\kappa), \tau) \) for \( \kappa \in [0, \kappa_{\rm cri}) \), indexed by \( a \in (0,1) \). Notice that direct computations show \( \partial c/\partial \kappa > 0 \) and \( \partial c/\partial a < 0 \). 
Also, $p(a,\kappa)$ diverges to infinity as $a\to0$. Since \( c(a,0) = 0 \), these curves originate from \( \{c=0\} \times \{p > 1+\tau\} \) and move upward as \( \kappa \) increases in the \( (p,c) \)-plane.  
As the intersection of Regimes I and II occurs at \( a = 1 \) (Remark~\ref{rem_doubly simply intersection}), the intersection of Regimes II and III occurs at \( \kappa = \kappa_{\rm cri} \). To prove our claim, it suffices to show that the critical line induced by \( \kappa = \kappa_{\rm cri} \) converges to \( p \to 0 \) as \( a \to 0 \).


Notice that as $a \to 0$, we have 
\begin{align*}
    H(a,\kappa) = (1-\tau)\Big(1-\big(\frac{\kappa}{1-\tau}\big)^2-2\frac{\kappa}{1-\tau} \log \frac{\kappa}{1-\tau}\Big)\frac{1}{a^2}+ 2(1- (\tau-\kappa)^2)\log a + O(1).
\end{align*}
Therefore it follows that as $ a \to 0$, 
$$
H(a, \kappa) \to \begin{cases}
+ \infty &\textup{if } \kappa <1-\tau,
\smallskip 
\\
-\infty &\textup{if } \kappa >1-\tau.
\end{cases}
$$ 
Here, we have used the fact that the function \( x \mapsto 1 - x^2 - 2x \log x \) changes sign only at \( x = 1 \).
Combining with the fact $\kappa_{\rm cri}\leq (1-\tau)(1-a^2)$, we obtain $\kappa_{\rm cri}\to (1-\tau)$ as $a\to 0$. Hence, we conclude that $p(a,\kappa)>0$ in Regime II and in particular, $p(a,\kappa_{\rm cri})\to 0$ as $a\to 0$. 
\end{rem}

\subsection{Electrostatic energies}\label{subsection simply connected energies}

In this section, we derive the weighted logarithmic energy for the simply connected regime. Recall that the Robin's constant is given by \eqref{eq:variational}. 

\begin{lem}\label{lem_energy simply connected}
Suppose that the droplet \( S \) is simply connected, with parameters \( a \in (0,1) \), \( \kappa \in [0, \kappa_{\rm cri}) \), and \( \tau \in [0,1) \), where \( c \) and \( p \) are given by~\eqref{eq_simply connected c} and~\eqref{eq_simply connected p}.
Then the Robin's constant, denoted $C_{\rm s}(p,c,\tau)$, is evaluated as
\begin{align}\label{eq_Robin simply}
    C_{\rm s}(p,c,\tau) &= 
    \frac{1+c}{2}- \frac{R\kappa p}{2a(1-\tau^2)}+c\log a - (1+c)\log R,
\end{align}
where $R$ is defined by~\eqref{eq:sc2}.
Furthermore, we have 
\begin{align}\label{eq_simply energy energy term}
\begin{split}
    \int_\mathbb{C} Q(\zeta) \ud \mu_Q(\zeta) &= \frac{1}{2}+2c-\frac{2cp^2}{1+\tau}+2c(1+2c)\log \frac{a}{R}+2c^2\log \frac{c(1-\tau^2)(1-a^2)}{\kappa}
    \\
    &\quad -\frac{R^3\kappa p}{(1-\tau^2)^2a^3}\Big((1+\tau a^2)\kappa-(1-\tau)(1-a^2)(1-\tau a^2)\Big)+\frac{(1-a^2)Rcp}{a(1+\tau)}-\frac{Rc\kappa p}{a(1-\tau^2)}.
\end{split}
\end{align}
In particular, $\mathcal{I}_{ \rm s}(p,c,\tau)$ is given by \eqref{weighted energy simply connected}.
\end{lem}

We mention that Lemma~\ref{lem_energy simply connected} extends previous results for \( \tau = 0 \) given in \cite[Lemma 4.8, pre-critical case]{BSY24}. 

\begin{proof}[Proof of  Lemma~\ref{lem_energy simply connected}] 
Note that by \eqref{eq_deriv of functional US} and~\eqref{eq_functional US derivative}, we have 
\begin{align}\label{eq_Schwarz asymp}
\begin{split}
    \mathsf{S}(\zeta)  = \tau \zeta + \frac{(1+c)(1-\tau^2)}{\zeta} + O(1/\zeta^2),  \qquad \zeta \to \infty. 
\end{split}
\end{align} 
Notice that the Robin's constant $C_{\rm s}(p,c,\tau)=\ell_S/2$, where $\ell_S$ is the constant value of $\mathcal{U}_S$ on $S$ defined as~\eqref{def of lK Robin 2}. Thus, it follows from \eqref{eq_functional computation 1} that 
\begin{align}\label{eq_Robin equation}
    \mathcal{U}_S(\zeta)- 2C_{\rm s}(p,c,\tau)  = \frac{1}{1-\tau^2}\Big(|\zeta|^2- \re \zeta\, \mathsf{S}(\zeta)- \re \int_{1/z}^z f(1/w)f'(w)\ud w\Big)
\end{align}
for $\zeta \notin S$ and $z= F(\zeta)$. By definition \eqref{def of functional US} of $\mathcal{U}_S$, the left hand side of \eqref{eq_Robin equation} has the asymptotic behaviour 
\begin{align*}
   \frac{1}{1-\tau^2}(|\zeta|^2-\re \zeta^2) - 2(1+c)\log|\zeta| -2C_{\rm s}(p,c,\tau) +O(1/\zeta), \qquad \zeta \to \infty.
\end{align*}
On the other hand, by using~\eqref{eq_Schwarz asymp} and~\eqref{eq_functional computation 2}, the right hand side has the asymptotic behaviour   
\begin{align*}
    \quad \frac{1}{1-\tau^2}(|\zeta|^2-\re \zeta^2) -(1+c) +\frac{R\kappa p}{a(1-\tau^2)} - 2c\log a -2(1+c)\log \frac{|\zeta|}{R} + O(1/\zeta), \qquad \zeta \to \infty.
\end{align*}
Comparing both sides of~\eqref{eq_Robin equation} at $\zeta \to \infty$, we obtain the desired identity \eqref{eq_Robin simply}.

Now we take $\zeta\to p$ and subsequently $z \to F(p)=1/a$ on both sides of~\eqref{eq_Robin equation}. Then the left hand side of~\eqref{eq_Robin equation} satisfies
\begin{align*}
    \frac{1}{1-\tau^2}\int_S \log \frac{1}{|\eta-p|^2}\ud A(\eta) + \frac{p^2}{1+\tau}-2c \log|\zeta-p|-2C_{\rm s}(p,c,\tau)+O(\zeta-p), \qquad \zeta \to p.
\end{align*}
Using~\eqref{eq_functional computation 2}, the right hand side of~\eqref{eq_Robin equation} satisfies
\begin{align*}
    \frac{(1-a^2)}{a(1+\tau)}Rp-\frac{R\kappa f'(1/a)}{a^2(1-\tau^2)}-2c\log |\zeta-p| + 2c\log \frac{(1-a^2)f'(1/a)}{a} + 2\log a+O(\zeta-p), \qquad \zeta \to p.
\end{align*}
Observe here that by \eqref{eq:sc3}, we have 
\begin{align*}
    f'(1/a) = R\Big(1-\tau a^2 + \frac{\kappa a^2}{(1-a^2)^2}\Big) = \frac{a^2c(1-\tau^2)}{R\kappa}.
\end{align*}
Then, comparing both sides of~\eqref{eq_Robin equation} we obtain 
\begin{align}\label{eq_simply energy log p computation}
\begin{split}
    &\quad \frac{1}{1-\tau^2}\int_S \log \frac{1}{|\eta-p|}\ud A(\eta) \\
    &= C_{\rm s}(p,c,\tau)-\frac{c}{2} -\frac{p^2}{2(1+\tau)}+\frac{(1-a^2)Rp}{2a(1+\tau)}+c\log\frac{ac(1-\tau^2)(1-a^2)}{R\kappa}+\log a\\
    &= \frac{1}{2}-\frac{p^2}{2(1+\tau)}-\frac{R\kappa p}{2a(1-\tau^2)}+\frac{(1-a^2)Rp}{2a(1+\tau)}+c\log\frac{c(1-\tau^2)(1-a^2)}{\kappa}+(1+2c)\log \frac{a}{R}.
\end{split}
\end{align}
Finally, from Green's formula and change of variables $\zeta = f(z)$
\begin{align*}
    \int_S |\zeta|^2-\tau \re \zeta^2 \ud A(\zeta) &= \frac{1}{2\pi i}\int_{\partial S}  \Big( \frac12 \bar{\zeta}-\tau \zeta \Big)|\zeta|^2  \ud \zeta= \frac{1}{2\pi i}\int_{\partial \mathbb{D}} \Big(\frac{1}{2}f(1/z)-\tau f(z)\Big)f(z)f(1/z)f'(z) \ud z.
\end{align*}
Then after straightforward computation evaluating residues at $z=0$ and $z=a$, we obtain
\begin{align}\label{eq_simply energy quadratic computation}
\begin{split}
    &\quad \int_S |\zeta|^2-\tau \re \zeta^2 \ud A(\zeta)
    \\
    &= \Big(c+\frac{1}{2}\Big)(1-\tau^2)^2-(1-\tau)(1-\tau^2)cp^2 -\frac{R^3\kappa p }{a^3}\Big((1+\tau a^2)\kappa-(1-\tau)(1-a^2)(1-\tau a^2)\Big).
\end{split}
\end{align}
Combining~\eqref{eq_simply energy log p computation} and~\eqref{eq_simply energy quadratic computation}, we conclude \eqref{eq_simply energy energy term}. 
Finally, the evaluation of $\mathcal{I}_{ \rm s }$ follows from \eqref{energy in terms of Robin}. 
\end{proof}

\section{Proofs of main results}  \label{Section_proofs}

This section culminates the results established in the previous sections and completes the proof of our main results.

\subsection{Proofs of Theorem~\ref{Thm_main droplet} and~\ref{Thm_droplet and energy}}  \label{Subsection_main droplet} 

We now summarise our results and highlight where each key ingredient of the proofs has been established.

\subsubsection{Doubly connected regime; Theorem~\ref{Thm_main droplet} (i) and  Theorem~\ref{Thm_droplet and energy} (i)}
 
In Proposition~\ref{prop_doubly connected droplet}, we established that the parameters \( (p, c, \tau) \) induce a doubly connected droplet if and only if \( D \subset E \), where \( D \) and \( E \) were defined in Section~\ref{Subsection_doubly connected droplet}.
In this case, the droplet is given by \( E \cap D^c \), and the weighted logarithmic energy is determined by~\eqref{weighted energy doubly connected}, as shown in Lemma~\ref{lem_energy doubly connected}.

It remains to verify that \( (p, c, \tau) \) lies in Regime I if and only if \( D \subset E \). This is an elementary computation, but we provide some details for the reader's convenience.
Suppose that \( c \) and \( \tau \) are fixed. If \( c(1-\tau^2) > (1+c)(1-\tau) \), then \( D \) cannot be contained in \( E \) for any \( p \).
Now, suppose \( c(1-\tau^2) \leq (1+c)(1-\tau) \). Note that the radius of curvature of \( \partial E \) at its rightmost point, \( (1+\tau)\sqrt{1+c} \), is given by 
\begin{align*}
    r_E:= \frac{(1-\tau)^2}{1+\tau}\sqrt{1+c}.
\end{align*}
As \( p \) increases from \( p = 0 \), the maximum value of \( p \) for which \( D \subset E \) holds is reached when \( \partial D \) and \( \partial E \) first become tangent.
If the radius of \( D \) is greater than \( r_E \), then \( \partial D \) and \( \partial E \) will be tangent at two conjugate points.
By eliminating $y$ in the algebraic equations of $\partial D$ and $\partial E$, we have 
\begin{align}\label{eq_ellipse quadratic equation}
    4\tau x^2-2(1+\tau)^2px +(1+\tau)^2p^2+(1+\tau)^2(1-\tau)(1-\tau-2c\tau) = 0.
\end{align}
Therefore, $D\subset E$ corresponds to the range of $p$ where discriminant of~\eqref{eq_ellipse quadratic equation} is not positive.
If the radius of $D$ is equal or smaller than $r_E$, $\partial D$ and $\partial E$ will meet at $(1+\tau)\sqrt{1+c}$, which is characterised by
\begin{align*}
    p+\sqrt{c(1-\tau^2)}\leq (1+\tau)\sqrt{1+c}.
\end{align*}
Combining above arguments, we have shown that $(p,c,\tau)$ falls within Regime I if and only if $D\subset E$.
Therefore, we complete the proof of Theorems~\ref{Thm_main droplet} (i) and~\ref{Thm_droplet and energy} (i).


\subsubsection{Simply connected regime; Theorem~\ref{Thm_main droplet} (ii) and  Theorem~~\ref{Thm_droplet and energy} (ii)}

By combining Propositions~\ref{prop:a priori}, ~\ref{prop_simply equality} and ~\ref{prop_simply inequality}, we have proven that the parameters \( (p, c, \tau) \) induce a simply connected droplet if and only if they fall within Regime II. Moreover, we have described the boundary of the droplet as the closure of the interior of the image of the unit circle under the rational map~\eqref{eq_simply conformal map}. Furthermore, the weighted logarithmic energy~\eqref{weighted energy simply connected} was established in Lemma~\ref{lem_energy simply connected}.

\subsubsection{Double component regime; Theorem~\ref{Thm_main droplet} (iii)}
To complete the proof, we verify that if the droplet consists of two disjoint simply connected components, then \( (p, c, \tau) \) belongs to Regime III. 
As shown in Section~\ref{Section_topology}, the only possible topology for the droplet in this case is two simply connected components, since the doubly connected case corresponds to Regime I and the simply connected case corresponds to Regime II.

\subsection{Proof of Corollary~\ref{Cor_moments}}  \label{Subsection_moments}

We now present the proof of Corollary~\ref{Cor_moments}. 
Recall that the partition functions $Z_N^{ \mathbb{C} }$ and $ Z_N^{ \mathbb{H} } $ are defined as normalisation constants in \eqref{Gibbs complex} and \eqref{Gibbs symplectic}. In both cases, it is well known that 
\begin{equation} \label{ZN expansions leading order}
\log Z_N^{ \mathbb{C} }(W) 
= -I_W( \mu_W ) N^2 +o(N^2) \qquad  \log Z_N^{ \mathbb{H} }(W) = -2\,I_W( \mu_W ) N^2 +o(N^2), 
\end{equation}
see e.g. \cite{BKS23} and references therein. 
Furthermore, it follows from \cite{AS21,Se23} that for the complex case, we have 
\begin{equation} \label{ZN expansion up to N}
\log Z_{N}^\C(W) = - I_W(\mu_W) N^2 +\frac12 N \log N + \Big( \frac{\log(2\pi)}{2}-1 - \frac12 \int_\C \log(\Delta W) \,\ud\mu_W \Big) N +o(N^{ \frac12 +\epsilon }) 
\end{equation}
for some $\epsilon >0$.  

We now connect the free energy expansions with the moments of the characteristic polynomials in Corollary~\ref{Cor_moments}. 
By their definitions, the moments of characteristic polynomials can be expressed in terms of the partition functions as 
\begin{equation} \label{moments in terms of ZN}
\mathbb{E} \Big[ \,  \Big|\det (X-z) \Big|^{2cN} \Big]  = \begin{cases}
Z_N^{ \mathbb{C} }(Q)/Z_N^{ \mathbb{C} }(W^{ \rm e } )  & \textup{for the complex case},
\smallskip 
\\
Z_N^{ \mathbb{H} }(Q)/Z_N^{ \mathbb{H} }(W^{ \rm e } )  & \textup{for the symplectic case},
\end{cases}
\end{equation}
where $Q$ is given by~\eqref{eq:potential} with $p=z$ and $W^{ \rm e }$ is given by \eqref{def of potential eGinibre}.
Indeed, by using the theory of planar orthogonal and skew-orthogonal polynomials (see e.g. \cite{BF24}), one can explicitly express $Z_N^{ \mathbb{C} }(W^{ \rm e } )$ and $Z_N^{ \mathbb{H} }(W^{ \rm e } )$ in terms of the Barnes $G$-function. 
It is also well known that 
\begin{equation} \label{energy for the eGinibre}
I_W( \mu_W^{ \rm e } ) = \frac34, 
\end{equation}
which can also be seen as \eqref{weighted energy doubly connected} with $c=0$. Then by combining Theorem~\ref{Thm_droplet and energy}, \eqref{ZN expansions leading order} (and also \eqref{ZN expansion up to N} for the complex case), \eqref{moments in terms of ZN} and \eqref{energy for the eGinibre}, we obtain the desired results. Here, for the complex case, we have used that 
\begin{align*}
\int_\C \log(\Delta Q) \,\ud\mu_Q  = \log \Big( \frac{1}{1-\tau^2}\Big),
\end{align*}
which follows from the fact that $\mu_Q $ has the total mass $1$.

\appendix

\section{Univalence criterion} \label{section univalence}

This appendix is devoted to the proof of Proposition~\ref{prop_univalence}. By definition, this reduces to finding the condition for a certain quadratic polynomial to have all its roots inside \( \bar{\mathbb{D}} \). Therefore, we present a specific case of the Schur-Cohn test that resolves this problem.  
Although the test is originally used to determine the number of roots of a polynomial of arbitrary degree within \( \mathbb{D} \), we focus on the quadratic case, which also accounts for roots on the boundary \( \partial \mathbb{D} \). For the general Schur-Cohn test, we refer to~\cite{Hen88}, and for its applications in quadrature domain theory, we refer to~\cite{AHT21}.

Let $p$ be a polynomial $p(w) = a_0 + a_1w + \ldots + a_nw^n \in \mathcal{P}_n$, where $\mathcal{P}_n$ denotes the set of complex coefficient polynomials of degree $\leq n$. The  reciprocal polynomial $p^\#(w)$ of $p \in \mathcal{P}_n$ is defined by
\begin{equation*}
    p^\#(w) = w^n\overline{p(1/\bar{w})} = \bar{a}_n +  \bar{a}_{n-1} w + \ldots +  \bar{a}_0 w^n.
\end{equation*}
The \textit{Schur transform} $S_n : \mathcal{P}_n \to \mathcal{P}_{n-1}$ is defined as
\begin{equation}
    S_n(p)(w) =  \bar{a}_0 p(w)- a_n p^\#(w).
\end{equation}
For $p \in \mathcal{P}_n$, we define
\begin{align*}
    p_0 = p, \qquad p_1= S_n(p_0), \qquad \ldots, \qquad  p_n= S_1(p_{n-1}).
\end{align*}

\begin{lem}\label{Schur-Cohn test}
Let $p\in \mathcal{P}_2$ with $p_1(0)<0$. Then all zeros of $p(w)$ lie in $\bar{\mathbb{D}}$ if and only if $p_2(0)\geq 0$.
\begin{proof}
Let $p(w)= a_0 + a_1w + a_2w^2 \in \mathcal{P}_2$. Then
\begin{align*}
    p_1(w) = |a_0|^2-|a_2|^2 + (\bar{a}_0a_1-a_2\bar{a}_1)w, \quad p_2(w) = (|a_0|^2-|a_2|^2)^2 - |\bar{a}_0a_1-a_2\bar{a}_1|^2.
\end{align*}
Then we have $a_2\neq 0$ since $p_1(0)<0$.

We denote by $w_1, w_2$ the two roots of $p$.
We first consider the case where a root of $p$ lies on $\partial \mathbb{D}$. Without loss of generality, set $|w_2|=1$. Then due to the condition $p_1(0)<0$, we have $|w_1|<1$. Furthermore,
\begin{align*}
    p_1(w)=|a_2|^2 (|w_1|^2-1)( 1 - \bar{w}_2w),
\end{align*}
which gives
\begin{align*}
    p_2(0) = |a_2|^4(|w_1|^2-1)^2(1-|w_2|^2) = 0.
\end{align*}

Next, assume that $p$ contains no zero on $\partial \mathbb{D}$. Then $p^\#$ also does not contain any zero on $\partial \mathbb{D}$. For $|w|=1$, the condition $p^\#(w)=w^n\overline{p(w)}$ implies
\begin{align*}
    |a_2p^{\#}(w)| > |\bar{a}_0p(w)|, \quad |w|=1.
\end{align*}
Since \( a_2 p^\# \) does not vanish on \( \partial \mathbb{D} \) by assumption, Rouché's theorem asserts that \( a_2 p^\# \) and \( p_1 = \bar{a}_0 p - a_2 p^\# \) have the same number of zeros in \( \mathbb{D} \).  
Note that the root of \( p_1 \) lies in \( \mathbb{D} \) if and only if \( p_2(0) < 0 \). Thus, the condition \( p_2(0) \geq 0 \) holds if and only if \( p^\# \) has no roots in \( \mathbb{D} \), which is equivalent to saying that all zeros of \( p \) are in \( \bar{\mathbb{D}} \).
\end{proof}
\end{lem}

\begin{proof}[Proof of Proposition~\ref{prop_univalence}]
Since univalence is preserved under translation and scalar multiplication, by \eqref{rel of f and g}, it suffices to consider the univalence of \( g \) in \eqref{def of rational g}.
By definition, $g$ is univalent on $\bar{\mathbb{D}}^c$ if for all $|z|>1$, all zeros of the quadratic polynomial
\begin{align}\label{eq:univ1}
\begin{split}
    p_z(w) &:= zw(z-a)(w-a) \frac{g(z)-g(w)}{z-w}
    = z(z-a)w^2 -\big((az+\tau)(z-a) - \kappa z\big)w + a\tau(z-a) 
\end{split}
\end{align}
lie in $\bar{\mathbb{D}}$. 
By the continuous dependence of the roots of \( p_z \) on \( z \), the function \( g \) is univalent on \( \bar{\mathbb{D}}^c \) if and only if all roots of \( p_z \) lie in \( \bar{\mathbb{D}} \) for \( |z| = 1 \).

Letting $p= p_z$, the Schur transforms $p_1 = S_2(p)$ and $p_2=S_1(p_1)$ in ~\eqref{eq:univ1} under $|z|=1$ are given by 
\begin{align}
    p_1(w) &= \Big(|z-a|^2(a(1-\tau^2) +\tau(1-a^2)z) - \kappa(1+\tau a^2) z + a\kappa(1+\tau)\Big)w - |z-a|^2(1-\tau^2a^2),\\
    p_2(w) &= |z-a|^4(1-\tau^2a^2)^2 - \Big||z-a|^2(a(1-\tau^2) +\tau(1-a^2)z) - \kappa(1+\tau a^2) z + a\kappa(1+\tau)\Big|^2. \label{eq:univ2}
\end{align}
Note that \( p_1(0) < 0 \) for all \( |z|=1 \). By virtue of Lemma~\ref{Schur-Cohn test}, it suffices to determine the range of \( \kappa \) for which \( p_2(0) \geq 0 \) for all \( |z| = 1 \). Expanding~\eqref{eq:univ2},  the condition \( p_2(0) \geq 0 \) is equivalent to 
\begin{align}\label{eq:univ3}
\begin{split}
&\quad  (1-a^2)(1-\tau^2)|z-a\tau|^2|z-a|^4 \geq |a(1+\tau)z - (1+\tau a^2)|^2\kappa^2
    \\
    &-2|z-a|^2 \Big(
    a(1-\tau^2)\big((1+\tau a^2)\re z - a(1+\tau)\big)
    + \tau (1-a^2)\big(1+\tau a^2- a(1+\tau)\re z\big)\Big)\kappa .   
\end{split}
\end{align}
The solution of~\eqref{eq:univ3} with respect to \( \kappa \) is given by a closed interval on the real line for each \( |z| = 1 \), since it is a quadratic inequality in \( \kappa \).  
Since our objective is to find the intersection of such closed intervals over all \( |z| = 1 \), the admissible range of \( \kappa \) is a single closed interval. Moreover, since \( \kappa = 0 \) satisfies the inequality~\eqref{eq:univ3} for all \( |z| = 1 \), the solution set is nonempty.

We first claim that \( \kappa_{\rm max} \), as defined in \eqref{eq_kappa max}, is the largest value of \( \kappa \) that satisfies~\eqref{eq:univ3} for all \( |z| = 1 \).  
Substituting \( \kappa = \kappa_{\rm max} \) into~\eqref{eq:univ3}, we obtain
\begin{align}
\begin{split}
    &\quad \frac{8\tau a^2(1+\tau)(1-a^2)}
    {(1+\tau a^2)^4}
    \Big((1+\tau a^2)\re z- a(1+\tau)\Big)^2\\
    &\times
    \Big( \big(1-
    \tau a^2 - 2\tau^2a^2 +2\tau a^4 +\tau^2 a^4 - \tau^3 a^6\big)
    -a(1-\tau)(1+\tau a^2)^2 \re z
    \Big) \geq 0, \label{eq:univ4}
\end{split}
\end{align}
For $|z|=1$, we have
\begin{align*}
    &\quad \big(1-
    \tau a^2 - 2\tau^2a^2 +2\tau a^4 +\tau^2 a^4 - \tau^3 a^6\big)-a(1-\tau)(1+\tau a^2)^2 \re z
    \\
    &\ge  (1-a)(1-\tau a)(1- \tau^2 a^4 + 2\tau a(1-a^2)) > 0. 
\end{align*}
Thus we have proven that $\kappa=\kappa_{\rm max}$ is admissible. Observe that the equality in~\eqref{eq:univ4} holds if and only if
\begin{equation*}
    z =\frac{a(1+\tau)}{1+\tau a^2}\pm
    \frac{\sqrt{(1-a^2)(1-\tau^2a^2)}}{1+\tau a^2}i .
\end{equation*}
Thus, if $\kappa> \kappa_{\rm max}$, the inequality~\eqref{eq:univ3} is violated at the same points. 

Lastly, we prove that \( \kappa_{\rm min} \) in \eqref{eq_kappa max} is the smallest value of \( \kappa \) that satisfies~\eqref{eq:univ3} for all \( |z| = 1 \).  
Again, substituting \( \kappa = \kappa_{\rm min} \) into~\eqref{eq:univ3} gives
\begin{align}\label{eq:univ5}
\begin{split}
    &\quad 4a(1-a)(1-\tau)(1-\re z)
    \\
    &\times\Big(2\tau a^2(1+\tau)(1+a)(\re z)^2 -2a(1+\tau)(1+\tau a)(1+\tau a^2)\re z
    \\
    &\qquad +(1+a^2+\tau a+2\tau a^2-\tau a^3-\tau^2a^2+2\tau^2a^3+\tau^2a^4+\tau^3a^3+\tau^3a^5)\Big) \geq0
\end{split}
\end{align}
for all $|z|=1$. Since
\begin{equation*}
    \frac{1}{2}  \frac{2a(1+\tau)(1+\tau a)(1+\tau a^2)}{2\tau a^2(1+\tau)(1+a)} = \frac{(1+\tau a)(1+ \tau a^2)}{2\tau a(1+a)} > 1,
\end{equation*}
it suffices to check the last term of inequality~\eqref{eq:univ5} when $z=1$. Indeed, one can notice that 
\begin{align*}
    &\quad2a^2\tau(1+a)(1+\tau) - 2a(1+\tau)(1+\tau a)(1+\tau a^2)
    \\
    &+(1+a^2+\tau a+2\tau a^2-\tau a^3-\tau^2a^2+2\tau^2a^3+\tau^2a^4+\tau^3a^3+\tau^3a^5) =(1-a)^2(1-\tau a)^2(1+\tau a) > 0.
\end{align*}
Therefore, the inequality~\eqref{eq:univ5} holds for \( \kappa = \kappa_{\rm min} \), with equality attained at \( z = 1 \).  
Again, if \( \kappa < \kappa_{\rm min} \), the inequality~\eqref{eq:univ3} would be violated at \( z = 1 \). Hence, the proof is complete.
\end{proof}



\end{document}